\documentclass[%
 reprint,
 amsmath,amssymb,
 aps,
]{revtex4-2}

\usepackage{amsthm}
\usepackage{mathrsfs}
\usepackage{systeme, mathtools}

\usepackage{graphicx}% Include figure files
\usepackage{dcolumn}% Align table columns on decimal point
\usepackage{bm}% bold math
\newtheorem{mytheo}{Theorem}

\newtheorem{mylemma}{Lemma}

\begin{document}

\title{Loop Quantum Cosmology of non-diagonal Bianchi models}

\author{Matteo Bruno}
 \email{matteo.bruno@uniroma1.it}
 \affiliation{Physics Department, Sapienza University of Rome, P.le A. Moro 5, 00185 Roma, Italy}
\author{Giovanni Montani}
 \email{giovanni.montani@enea.it}
 \affiliation{ENEA, C.R. Frascati (Rome), Italy Via E.\ Fermi 45, 00044 Frascati (Roma), Italy} 
 \affiliation{Physics Department, Sapienza University of Rome, P.le A. Moro 5, 00185 Roma, Italy}

\begin{abstract}
 The non-diagonal Bianchi models are studied in the loop framework for their classical and quantum formulation. The expressions of the Ashtekar-Barbero-Immirzi variables and their properties are found to provide a loop quantization of these models. In the special case of Bianchi I Universe, it is shown that the geometrical operators result invariant from the diagonal description. Hence, the kinematical Hilbert space of the non-diagonal Bianchi I model has similar features to the diagonal one.
\end{abstract}

%\keywords{Suggested keywords}%Use showkeys class option if keyword
                              %display desired
\maketitle

%\tableofcontents
\section{Introduction}
One of the most natural arenas to test quantum gravity proposals \cite{thiemann_2007,cianfrani2014canonical} is provided by the Bianchi Universes \cite{BKL70, Montani2008}, in particular, by the simplest Bianchi I model \cite{KL63,thorne1973gravitation}, which is characterized by a zero spatial curvature.\\
The canonical quantization of the Bianchi I Universe has been faced in the metric approach in \cite{Misner1969,Benini_2007} and its classical singular nature has been preserved via the quantum dynamics, as soon as the behaviour of localized wave packet is concerned (for the possibility to implement a quantum bounce see \cite{Giovannetti2022}).\\
The implementation of Loop Quantum Gravity to the Bianchi I dynamics has been pursued in \cite{ashtekar2003mathematical,ashtekar2006quantumI,ashtekar2006quantumII,ashtekar2009loop} and a bouncing cosmology picture emerged as a consequence of the discrete nature of the geometrical operator spectrum. This procedure of quantization has been criticized in \cite{Cianfrani_2012}, see also \cite{CianfraniMontani2012critical,CianfraniMontani2012gauge}, comparing the structure of the $SU(2)$ symmetry of 
the general case with the theory emerging from the homogeneity constraint. The same bouncing behaviour of the Bianchi I cosmology has been also observed in the metric approach when the polymer quantum mechanics is implemented in the cosmological configurational space \cite{MONTANI2018,Moriconi2016,Moriconi2017,ANTONINI2019,Giovannetti2021}.\\

However, the standard formulation of the Bianchi models which is addressed for quantum analyses is the so-called diagonal case, i.e. when the three independent directions of space are governed by three different scale factors and no off-diagonal contributions emerge (in this case the super-momentum constraint results to be identically vanishing). A different, non-diagonal, approach to classical Bianchi cosmologies has been discussed in \cite{RYAN1972}, see also \cite{Belinski_2014}, where the implications of dealing with off-diagonal terms have been extensively studied both from the point of view of a Hamiltonian and field equation formulation. These studies clarify how the non-diagonal Bianchi Universes have a more complicated configurational structure, in which additional infinite walls came out in the standard spatial curvature terms when the singularity is approached.\\
An interesting generalization of the Bianchi models, as viewed in the framework of the Loop quantization, in correspondence to a non-zero Gauss constraint has been discussed in \cite{bojowald2000loopI,bojowald2000loopII,bojowald2000loopIII,bojowald2013mathematical}, defining a kinetical construction for the Hilbert space of the theory which accounts for the real $SU(2)$ symmetry of the general formulation.\\

Here, we considered an intermediate scenario for the quantization of the Bianchi I cosmology, never addressed before, which consists of starting from a metric formulation of a non-diagonal Bianchi I model, up to arriving at a kinematical formulation of the quantization procedure in terms of Ashtekar-Barbero-Immirzi variables.\\
We start with a metric representation similar to the one in \cite{RYAN1972,Belinski_2014} and, then, we construct first the standard Hamiltonian formulation. Subsequently, we translate our analysis in terms of the Ashtekar-Barbero-Immirzi variables.\\
The peculiarity of our representation is that we deal with three diagonal connection-like variables and three Euler angles, accounting for the non-diagonal nature of the model, i.e. for the rotation in time of the spatial directions, along which the connections are referred. It is worth noting that also in this non-diagonal setting, the Gauss constraint identically vanishes so that the $SU(2)$ symmetry can not be regarded as the driving structure in constructing the kinematical Hilbert space.\\

As one of the main results of the present analysis, we are able to demonstrate that, via a suitable rotation, we can lead back to a diagonal representation of the fluxes which appears as isomorphic to that one of the diagonal case. This fact suggests that the quantization procedure can be separated into two different parts, one corresponding to three diagonal connections and the other one involving the three angles.\\ 
This point of view is enforced by the possibility to introduce a positive definite scalar product at fixed values of the angle in a given state. So, we arrive at defining a kinematical Hilbert space for the model, which resembles the Bohr compactification procedure for the connection-like variables and a standard orthonormality request for states, corresponding to different angles.\\ 

Then, in the second part of the manuscript, we search for a more axiomatic formulation of the kinematical Hilbert space, based on a $U(1)^3$ representation of the fundamental quantization algebra and an extension to a $U(1)^6$ representation. Both approaches appear to be not completely viable but offer an interesting theoretical and physical point of view on how the quantum features of the non-diagonal Bianchi I model can emerge starting from different settings of the configurational quantum space.\\
The $U(1)^3$ formulation applies the Bohr compactification procedure to the connection-like variables only but finds the difficulty that the angles unavoidable enter in the almost-periodic functions, so preventing a rigorous and conclusive construction of a kinematical Hilbert space.\\
The analysis based on the symmetry $U(1)^6$ is instead constructed by a ``natural" extension of the almost-periodic functions' space to match it with the correct numbers of independent variables, all the six configurational coordinates are associated with a Bohr compactification procedure.\\
This formulation emulates that one in \cite{bojowald2013mathematical}, replacing the $SU(2)$ with the $U(1)^3$ and $U(1)^6$ symmetry. However, the different nature of the quantum numbers emerging in the two approaches prevents a complete parallelism, which could provide a natural kinematical Hilbert space.\\ 

We conclude that the physical relevance of quantizing a non-diagonal Bianchi I model relies on the classical notion that, such a morphology is naturally induced when the matter is introduced, for instance in the form of a perfect fluid not at rest with the synchronous reference frame.\\
On a classical level, the effect of the matter has been shown to induce a ``slow" rotation of the so-called Kasner axes, i.e. the independent directions to which the 
scale factors are referred. This rotation effect is not able to alter some properties of the Bianchi models, like the chaotic nature of the Bianchi VIII and IX cosmologies near enough to the singularity.\\
To some extent, the present construction of a kinematical Hilbert space, in which the angles are not affected by a Bohr compactification scenario and they trivially enters the scalar product between two states, can be regarded as a pre-dynamical construction which is inspired by the classical behaviour of a non-diagonal model, discussed above.\\

The paper is structured as follows.\\
In Sec.\ref{sec:Review} we recap some mathematical aspects of the classical approach to Bianchi models. Moreover, we present the loop quantization procedure for the diagonal Bianchi I model that will be the base for the quantization of the non-diagonal one.\\
In Sec.\ref{sec:non-d} we present some calculations of classical quantities for the non-diagonal Bianchi models. In particular, we find the ADM Lagrangian and the Ashtekar-Barbero variables for any non-diagonal Bianchi universes.\\
In Sec.\ref{sec:Bianchi1} we restrict our classical analysis to Bianchi I models. We compute the constraints and verify the good properties of their algebra.\\
In Sec.\ref{sec:quan1} we study the features of the holonomy in non-diagonal Bianchi models. Furthermore, we propose a quantization based on the quantum geometry, that allows us to quantize the theory in a similar way to the diagonal case.\\
In Sec.\ref{sec:repp} we want to analyze an alternative approach to quantize the theory, more similar to the canonical Loop Quantum Gravity. It is based on the representation of the holonomy group, which, in our setting, is commutative. We use the groups $U(1)^3$, in analogy with the other works on Loop Quantum Cosmology, and $U(1)^6$. Despite the construction of the kinematical Hilbert space is induced from the one for $SU(2)$, these theories have some issues in the physical interpretations.   

\section{Homogeneous universes and Bianchi models \label{sec:Review}}
In this section, we will summarize the known aspects of the Bianchi models and their quantization, with particular attention to the diagonal anisotropic Bianchi I model.\\
In a homogeneous model, the space-time is a manifold $\mathcal{M}=\mathbb{R}\times\Sigma$, where $\Sigma$ is a three-dimensional homogeneous space. We also require that the group of isometries $S$ acts freely on $\Sigma$, thus $\Sigma$ can be identified with $S$. On such a space, exists a basis of left-invariant one-form $\omega^I$ that satisfies the Maurer-Cartan equation
\begin{equation}
    \label{MC-eq}
    d\omega^I+\frac{1}{2}f^I_{JK}\omega^J\wedge\omega^K=0
\end{equation}
Moreover, the left-invariant vector fields on $\Sigma$ define a Lie algebra $\mathfrak{s}$. A basis for this algebra is given by the dual of the $\omega^I$, such vectors $\xi_I$ satisfy $\omega^J(\xi_I)=\delta^J_I$, $\theta_{MC}=\sum_I \omega^I\xi_I$, where $\theta_{MC}$ in the Maurer-Cartan one-form on $S$, and
\begin{equation}
    \label{g-eq}
    [\xi_I,\xi_J]=f^K_{IJ}\xi_K.
\end{equation}
The Riemannian metric $h$ induced on $\Sigma$ by the space-time metric $g$ can be decomposed as
\begin{equation}
    h=\eta_{IJ}\omega^I\otimes\omega^J
\end{equation}
where $\eta_{IJ}$ is a constant symmetric tensor on $\Sigma$.\\
Moreover, the homogeneous connection $A$ on $\Sigma$ can be written as $A=\phi\circ\theta_{MC}$, where $\phi:\mathfrak{s}\to\mathfrak{su}(2)$ is a linear map \cite{bojowald2013mathematical}.\\
Using coordinates $(t,x^i)$ adapted to the ADM formalism, where $x^i$ are a set of coordinates on $\Sigma$, the Ashtekar's variables can be written as
\begin{equation}
    \label{MaBoVa}
    A^a_i=\phi^a_I\omega^I_i,\ \ \ \ E^i_a=|det(\omega^I_a)|p^I_a\xi^i_I.
\end{equation}

Bianchi I models are characterized by null structure constants $f^K_{IJ}=0$. In this case, $S=\mathbb{R}^3$ and the space-time is topologically $\mathbb{R}^4$. Moreover, $\Sigma$ is flat;, in fact, it exists a set of coordinates in which $\omega^I_i=\delta^I_i$ (since $[\xi_I,\xi_J]=0$), thus $h_{i,j}(t,x)=\eta_{IJ}\delta^I_i\delta^J_j$ is constant on $\Sigma$.

\subsection{Loop quantization of the Bianchi I universe \label{sec:loop1}}
The diagonal Bianchi I model allows a quantization of the universe in the loop quantization program. This quantization was provided by A.Ashtekar and E. Wilson-Ewing in \cite{ashtekar2009loop}. The tensor $\eta_{IJ}$ is diagonal in diagonal models, such as the connection and the densitized triads. The metric can be written as
\begin{equation}
    h=-Ndt^2+a_1^2(dx^1)^2+a_2^2(dx^2)^2+a_3^2(dx^3)^2,
\end{equation}
where $a_i$ are the scale factors in each direction. While the Ashtekar's variables read
\begin{equation}
    A^a_i=c^a(L^a)^{-1}\omega^a_i,\ \ \ \ E^i_a=|det(\omega^I_a)|L_aV_0=^{-1}p_a\xi^i_a.
\end{equation}
These variables are fully characterized by $c^a,p_a$, and the following Poisson brackets hold
\begin{equation}
    \{c^a,p_b\}=\frac{8\pi G\gamma}{c^3 V_0}\delta^a_b,
\end{equation}
where $\gamma$ is the Barbero-Immirzi parameter, $V_0$ is the fiducial volume of the fiducial cell $\mathcal{V}$ (i.e. $V_0=\int_{\mathcal{V}}\omega^1\wedge\omega^2\wedge\omega^3=L_1L_2L_3$, where $L_I$ is the fiducial length of the I-th edge).
The diagonal components of the Ashtekar's variables can be written in terms of scale factors
\begin{align}
    \label{D-Va}
    p_I&=\mathrm{sgn}(a_I)a_Ja_KL_JL_K\ \mathrm{with}\ \epsilon_{IJK}=1,\\
    c_I&=\frac{\gamma L_I}{N}\frac{da_I}{dt}
\end{align}

The quantum theory presents basis states $|p_1,p_2,p_3\rangle$, which are eigenstates of the quantum geometry. In the state $|p_1,p_2,p_3\rangle$ the face $\sigma_i$ of the fiducial cell $\mathcal{V}$ orthogonal to the axis $x^i$ has area $|p_i|$. Moreover, the elementary operators act on this basis as 
\begin{align}
    &\hat{p}_1|p_1,p_2,p_3\rangle=p_1|p_1,p_2,p_3\rangle,\\ &\widehat{\exp(i\lambda c)}|p_1,p_2,p_3\rangle=|p_1-k\gamma\hbar\lambda,p_2,p_3\rangle
\end{align}
where $k=\frac{8\pi G}{c^3}$ is the gravitational constant.\\
From the quantum geometry in LQG, we know that, to have the best coarse grained, the state $|p_1,p_2,p_3\rangle$ is reproduced by an LQG state associated with a spin network that intersects the surface $\sigma_3$ with $N_3$ edges, each one carrying a quantum of area $4\pi\gamma\sqrt{3}\ell_P^2$. Hence, $N_3$ is given by
\begin{equation}
    4\pi\gamma\sqrt{3}\ell_P^2N_3=|p_3|
\end{equation}
Consider the rectangle of minimal area pierced by exactly one edge. We refer to it as the plaquette $\square_{1,2}$ whose area is $4\pi\gamma\sqrt{3}\ell_P^2$. The fiducial length of its edges is $\bar{\mu}_1L_1$ and $\bar{\mu}_2L_2$. Since the fiducial area on the surface is $L_1L_2$, we obtain
\begin{equation}
    N_3\bar{\mu}_1L_1\bar{\mu}_2L_2=L_1L_2.
\end{equation}
Equaling the two previous equations, we have
\begin{equation}
    \bar{\mu}_1\bar{\mu}_2=4\pi\gamma\sqrt{3}\frac{\ell_P^2}{|p_3|}
\end{equation}
Repeating this procedure for the faces $\sigma_1$ and $\sigma_2$, we can characterize the $\bar{\mu}_i$
\begin{equation}
    \bar{\mu}_i=\sqrt{4\pi\gamma\sqrt{3}\frac{|p_i|}{|p_jp_k|}}\ \ \ \mathrm{with}\ \epsilon_{ijk}=1.
\end{equation}
Thus, the quantization of the geometry characterizes the states of the kinematical Hilbert space. Moreover, it gives us a natural choice of the plaquette for calculating the curvature operator, which is necessary for describing the quantum dynamics.

\section{Non-diagonal models in the metric variables \label{sec:non-d}}
In a non-diagonal Bianchi model, the hypersurface $\Sigma$ has a metric
\begin{equation}
\label{mymetr}
    h_{ij}(t,x)=\eta_{IJ}(t)\omega^I_i(x)\omega^J_j(x).
\end{equation}
Following Belinski's approach to non-diagonal metric in \cite{Belinski_2014}, we want to factorize $\eta_{IJ}$ in a diagonal matrix similar to it via a rotation. Since $\eta_{IJ}(t)$ is a symmetric matrix depending $C^{\infty}$ on one parameter, it always has three eigenvalues $\{a^2(t),b^2(t),c^2(t)\}$, that can be chosen continuously differentiable on the whole interval \cite{rellich1969perturbation}. Moreover, if the continuous eigenvalues are such that no two of them are equal at any $t\in \mathbb{R}$ if they are not equal for all $t\in \mathrm{R}$. Then all the eigenvalues and all the eigenvectors can be chosen $C^{\infty}$ in $t$ \cite{alekseevsky1998choosing}. Hence, $\eta_{IJ}$ can be decomposed as
\begin{equation}
\label{mydecom}
    \eta_{IJ}=\Gamma_{AB}R^A_I R^B_J\ \ \ \textrm{with}\ \ \Gamma_{AB}=\begin{pmatrix}
                a^2 & 0 & 0\\
                0 & b^2 & 0\\
                0 & 0 & c^2
                \end{pmatrix}
\end{equation}
and $R$ a rotation matrix depend on the parameter $t$ and it is infinitely differentiable. The eigenvalues are positive because $\eta_{IJ}$ are the components of the Riemannian metric $h$ in basis $\omega^I$, therefore, it is a positive definite, symmetric matrix.\\ 
The rotation introduces three new variables: the angles. Using the Euler angles $\{\theta,\psi,\varphi\}$ the rotation reads
\begin{equation}
    R=\exp(\theta j_3)\exp(\psi j_2)\exp(\varphi j_3)
\end{equation}
where $j_2$ and $j_3$ are the generators of $SO(3)$, as matrices:
\begin{equation*}
    j_2=\begin{pmatrix}
            0 & 0 & -1\\
            0 & 0 & 0\\
            1 & 0 & 0
        \end{pmatrix}
    \ \ \ \ \ \ \ j_3=\begin{pmatrix}
            0 & 1 & 0\\
            -1 & 0 & 0\\
            0 & 0 & 0
        \end{pmatrix}
\end{equation*}
These angles can be interpreted as a physical rotation of the left-invariant vector fields and they are $C^{\infty}$ functions of $t$.

\subsection{Calculation of the ADM Lagrangian \label{sec:ADM-lagrangian}}
To construct a Hamiltonian theory, one can choose as configuration variables the three scale factors and the three Euler angles $\{a,b,c,\theta,\psi,\varphi\}$, from now, these variables are called metric variables. The Lagrangian can be calculated from the ADM formalism 
\begin{align}
\nonumber
    L_{ADM}=N\sqrt{h}\Big(\bar R +&\frac{1}{4N^2}h^{ri}h^{sj}(\dot h_{ij}-N_{i;j}-N_{j;i})\\ \nonumber
    &\ \ \ \ \ \times(\dot h_{rs}-N_{r;s}-N_{s;r})+\\ \nonumber
    -&\frac{1}{4N^2}h^{ij}h^{rs}(\dot h_{ij}-N_{i;j}-N_{j;i})\\
    &\ \ \ \ \ \times(\dot h_{rs}-N_{r;s}-N_{s;r})\Big).
\end{align}
A sketch of the calculation is presented below and passes through four terms. We rewrite the Lagrangian as the sum of four terms (excluding the scalar curvature $\Bar{R}$)
\begin{align*}
     L_{ADM}&=N\sqrt{h}\bigg(\bar R +\frac{1}{4N^2}\Big(h^{ri}h^{sj}\dot h_{ij}\dot h_{rs}\\
     &+h^{ri}h^{sj}(N_{i;j}+N_{j;i})(N_{r;s}+N_{s;r})\\
    &-h^{ri}h^{sj}\dot h_{ij}(N_{r;s}+N_{s;r})-h^{ri}h^{sj}\dot h_{rs}(N_{i;j}+N_{j;i})\\
    &-h^{ij}h^{rs}(\dot h_{ij}-N_{i;j}-N_{j;i})(\dot h_{rs}-N_{r;s}-N_{s;r})\Big)\bigg).
\end{align*} 
For the first term, we simply factorize the metric $h_{ij}$ as in Eq.(\ref{mymetr}) and (\ref{mydecom}).
\begin{align}
\label{1term}
\nonumber
    h^{ri}h^{sj}\dot h_{ij}\dot h_{rs}&=\Gamma^{dc}\Gamma^{gh}(\dot{\Gamma}_{dh}+\Gamma_{a'h}\dot R^{a'}_b\Lambda^b_d+\Gamma_{db'}\dot R^{b'}_f\Lambda^f_h)\\ \nonumber
    &\ \ \ \ \ \ \ \ \ \ \ \ \times(\dot{\Gamma}_{cg}+\Gamma_{g'g}\dot R^{g'}_a\Lambda^a_c+\Gamma_{ch'}\dot R^{h'}_e\Lambda^e_g)\\ \nonumber
    &=\Gamma^{ac}\Gamma^{bd}\dot{\Gamma}_{ab}\dot{\Gamma}_{cd}-4\Gamma^{ac}\dot{\Gamma_{ab}}(R\dot\Lambda)^b_c\\
    &\ \ \ \ \ +2\Gamma^{ab}\Gamma_{cd}(R\dot\Lambda)^d_a(R\dot\Lambda)^c_b+2(R\dot\Lambda)^b_c(R\dot\Lambda)^c_b
\end{align}
where $\Gamma^{ab}$ is the inverse of $\Gamma_{ab}$ and $\Lambda$ is the inverse matrix of $R$. Recalling that, since $R$ is a rotation, $\Lambda=R^t$ , therefore, the following properties hold
\begin{align}
\label{Rprop}
\nonumber
    &(a)\ \ \dot R \Lambda+R\dot \Lambda=\Lambda\dot R+\dot \Lambda R=0\\
    &(b)\ \ (R\dot \Lambda)^t=\dot R \Lambda=-R\dot\Lambda\\ \nonumber
    &(c)\ \ tr(R\dot\Lambda)=0
\end{align}
The property $(\ref{Rprop}a)$ follows from the Leibniz rule. The $(\ref{Rprop}b)$ can be proved using $(\ref{Rprop}a)$ and it shows the skew-symmetry of $R\dot\Lambda$, so it implies $(\ref{Rprop}c)$. Thus, in (\ref{1term}) one term vanishes and the first term reads
\begin{align*}
    h^{ri}h^{sj}\dot h_{ij}\dot h_{rs}=&\Gamma^{AC}\Gamma^{BD}\dot{\Gamma}_{AB}\dot{\Gamma}_{CD}\\
    &+2\Gamma^{AB}\Gamma_{CD}(R\dot\Lambda)^D_A(R\dot\Lambda)^C_B\\
    &+2(R\dot\Lambda)^B_C(R\dot\Lambda)^C_B.
\end{align*}
 For the other terms, it is useful to recall the Maurer-Cartan equation in local coordinates
\begin{equation}
    \frac{\partial\omega_i^I}{\partial x^j}-\frac{\partial\omega_j^I}{\partial x^i}= f^I_{JK}\omega_i^J\omega_j^K.
\end{equation}
We are interested to compute a common factor to all the remaining terms $(N_{i;j}+N_{j;i})$.\\
Before this, we need to compute the Christoffel symbols
\begin{align}
\label{Christoffel}
\nonumber
    \bar{\Gamma}^k_{ij}N_k&=\frac{1}{2}h^{kl}(h_{il,j}+h_{jl,i}-h_{ij,l})N_k\\ \nonumber
    &=\frac{1}{2}(N_{i,j}+N_{j,i})+\frac{1}{2}N^l\eta_{IJ}\big(\omega^I_i\omega^J_{l,j}\\
    &\ \ \ \ \ \ +\omega^I_j\omega^J_{l,i}-(\omega^I_i\omega^J_j)_{,l}\big).
\end{align}
After some tedious calculations, the structure constants emerge using Eq.(\ref{MC-eq}). In fact, the factor reads
\begin{align*}
    (N_{i;j}&+N_{j;i})=(N_{i,j}+N_{j,i}-2\Bar{\Gamma}^k_{ij}N_k)\\
    &=-N^l\eta_{IJ}\big(\omega^I_i\omega^J_{l,j}+\omega^I_j\omega^J_{l,i}-(\omega^I_i\omega^J_j)_{,l}\big)\\
    &=-N^l\Big(\eta_{IJ}\omega^I_i(\omega^J_{l,j}-\omega^J_{j,l})+\eta_{IJ}(\omega^I_j\omega^J_{l,i}-\omega^I_{i,l}\omega^J_j)\Big).
\end{align*}
For symmetry of $\eta_{IJ}$ and the Maurer-Cartan equation, we get
\begin{align*}
    (N_{i;j}+N_{j;i})&=-N^l\Big(\eta_{IJ}\omega^I_if^J_{KL}\omega^K_l\omega^L_j+2\eta_{IJ}\omega^I_j\omega^J_{[l,i]}\Big)\\
    &=2N^K\eta_{IJ}f^J_{KL}\omega^I_{(i}\omega^L_{j)}.
\end{align*}
We recall that the shift vector $N^i$ on a homogeneous space can be factorized as $N^i(t,x)=N^I(t)\xi^i_I(x)$, where $N^I=N^i\omega_i^I$ depends on the time only.\\
For the second term, we find 
\begin{align}
\nonumber
    h^{ri}h^{sj}&(N_{i;j}+N_{j;i})(N_{r;s}+N_{s;r})=\\ 
    =&2N^A N^B(f^I_{AJ}f^J_{BI}+\eta^{IJ}\eta_{KL}f^K_{AI}f^L_{BJ})
\end{align}
The third term is the double product, it is trivial that 
$$h^{ri}h^{sj}\dot h_{ij}(N_{r;s}+N_{s;r})=h^{ri}h^{sj}(N_{i;j}+N_{j;i})\dot h_{rs}$$
since one can obtain one from another, via renaming $i\leftrightarrow r,\,j\leftrightarrow s$. This term can be easily calculated in a similar way to the second term
\begin{equation}
    h^{ri}h^{sj}\dot h_{ij}(N_{r;s}+N_{s;r})=-2N^K\eta^{IJ}\dot{\eta}_{JL}f^L_{KI}.
\end{equation}
The fourth term appears quadratically in the Lagrangian, it is
\begin{equation}
    h^{ij}(\dot h_{ij}-N_{i;j}-N_{j;i})=\eta^{IJ}\dot{\eta}_{IJ}+2N^K f^J_{KJ}.
\end{equation}
Recalling that for Bianchi A class models $f^J_{KJ}=0$, this term does not depend on the shift vector. To complete the decomposition, it remains to write $\eta_{IJ}$ and its inverse $\eta^{IJ}$ in terms of diagonal and rotation matrices. As matrix $\eta=\Lambda\Gamma R$, then,
$$\eta^{-1}\dot{\eta}=\Lambda\Gamma^{-1}R\dot{\Lambda}\Gamma R+\Lambda\Gamma^{-1}\dot{\Gamma} R+\Lambda\dot R,$$
and, from properties (\ref{Rprop}), $tr(\eta^{-1}\dot{\eta})=tr(\Gamma^{-1}\dot{\Gamma})$.\\
The ADM Lagrangian can be written in terms of scale factors and rotation matrices
\begin{widetext}
\begin{align}
\label{non-diagL}
\nonumber
    L_{ADM}=&N|\mathrm{det}(\omega^I_i)|\sqrt{\mathrm{det}(\Gamma_{AB})}\bigg(\bar R+\frac{1}{4N^2}\Big(\Gamma^{AC}\Gamma^{BD}\dot{\Gamma}_{AB}\dot{\Gamma}_{CD}+2\Gamma^{AB}\Gamma_{CD}(R\dot\Lambda)^D_A(R\dot\Lambda)^C_B\\
    &+2(R\dot\Lambda)^B_C(R\dot\Lambda)^C_B+2N^A N^B(f^I_{AJ}f^J_{BI}+\eta^{IJ}\eta_{KL}f^K_{AI}f^L_{BJ})+4N^K\eta^{IJ}\dot{\eta}_{JL}f^L_{KI}-\Gamma^{IJ}\dot{\Gamma}_{IJ}\Gamma^{KL}\dot{\Gamma}_{KL}\Big)\bigg).
\end{align}
\end{widetext}
The Lagrangian can be further manipulated. First of all, the scalar curvature can be written in terms of structure constant and $\eta$ only \cite{Landau2.14}. Moreover, $\Gamma$ can be written explicitly as a diagonal matrix $\Gamma_{AB}=a^2_{(A)}\delta_{AB}$, where $a_1=a,\,a_2=b,\,a_3=c$ are the scale factors, hence, $\sqrt{\mathrm{det}(\Gamma_{AB})}=|abc|$.\\
Furthermore, we can bring out $|\mathrm{det}(\omega^I_i)|$ from the integration on the space-time to compute the Einstein-Hilbert action $\mathcal{S}_{EH}$, because it is the only term that depends on the point of the hypersurface. It means
\begin{align}
\nonumber
    \mathcal{S}_{EH}&=-\frac{c^3}{16\pi G}\int dt\, d^3x\, L_{ADM}(t,x)\\ \nonumber
    &=-\frac{c^3}{16\pi G}\int dt\, \mathcal{L}(t)\int_{\Sigma}d^3x\,|\mathrm{det}(\omega^I_i)|\\
    &=-\frac{c^3}{16\pi G}V_0\int dt\, \mathcal{L}(t).
\end{align}
Hence, one can consider a Lagrangian that depends only on time $\mathcal{L}(t)$, which is defined as the homogeneous part of the Lagrangian defined in (\ref{non-diagL}).\\
The Lagrangian is complex and its writing in terms of elementary functions is quite difficult to read, an immediate simplification is to consider the Bianchi I model in which the structure constants vanish. Notice that, in this case, all terms which contain the shift vector vanish.

\subsection{Ashtekar's variables \label{sec:Ashtekar}}
The loop approach requires computing the densitized dreibein and the Ashtekar-Barbero variables. The dreibein vectors are pretty simple and they follow from the metric
\begin{equation}
    h_{ij}=\eta_{IJ}\omega^I_i\omega^J_j=\delta_{ab}e^a_ie^b_j
\end{equation}
using the decomposition of $\eta_{IJ}$ in Eq.(\ref{mydecom}), it easy to check
\begin{equation}
    e^a_i=a_{(a)}R^a_I\omega^I_i\ \ \ \ \ \textrm{$a$ is not summed,}
\end{equation}
where $a_1=a,a_2=b,a_3=c$. Its dual reads
\begin{equation}
    e^i_a=\tfrac{1}{a_{(a)}}\Lambda^I_a\xi^i_I\ \ \ \ \ \textrm{$a$ is not summed.}
\end{equation}
The proof that one is the dual of the other is trivial
\begin{equation*}
    e^a_ie^i_b=\tfrac{a_{(a)}}{a_{(b)}}R^a_I\omega^I_i\Lambda^J_b\xi^i_J=\tfrac{a_{(a)}}{a_{(b)}}R^a_I\Lambda^I_b=\tfrac{a_{(a)}}{a_{(b)}}\delta^a_b=\delta^a_b.
\end{equation*}
From this, the densitized dreibein is defined as
\begin{align}
\nonumber
    E^i_a=&\sqrt{h}e^i_a=|\mathrm{det}(\omega^I_i)|\frac{|a_1a_2a_3|}{a_{(a)}}\Lambda^I_a\xi^i_I\\
    =&|\mathrm{det}(\omega^I_i)|\mathrm{sgn}(a_{(a)})|a_ba_c|\Lambda^I_a\xi^i_I\ \ \ \ \textrm{with $\epsilon_{abc}=1$.}
\end{align}
Notice that $\Lambda$ is not a pure gauge rotation but it is a physical rotation, applied on the left-invariant vectors, necessary to have a non-diagonal metric.\\
To compute the connection we use the definition $A^a_i=\Gamma^a_i+\gamma K^a_i$. First of all, the extrinsic curvature is calculated
\begin{align}
\nonumber
    &K^a_i=K_{ij}e^{ja}=\frac{1}{2N}(\dot h_{ij}-N_{i;j}-N_{j;i})h^{jk}e^a_k\\ 
    &=\frac{1}{2N}a_{(a)}R^a_L\big(\eta^{LJ}\dot{\eta}_{JI}+N^A\eta^{LK}\eta_{IJ}f^J_{AK}+N^Af^L_{AI}\big)\omega^I_i.
\end{align}
For the spin part $\Gamma^a_i=\frac{1}{2}\epsilon^a_{\ bc}\omega^{bc}_i$, the spin connection is evaluated
\begin{align}
\nonumber
    \omega^{ab}_i=\frac{1}{2}\tfrac{a_{(b)}}{a_{(a)}}\Lambda^I_aR^b_Jf^J_{IK}\omega^K_i-\frac{1}{2}\tfrac{a_{(a)}}{a_{(b)}}\Lambda^I_bR^a_Jf^J_{IK}\omega^K_i+\\
    -\frac{1}{2}\tfrac{1}{a_{(a)}a_{(b)}}\eta_{LK}\Lambda^I_a\Lambda^J_b f^L_{JI}\omega^K_i.
\end{align}
One can check that this expression is skew-symmetric in $a,b$. Hence, for the spin term we obtain
\begin{equation}
    \Gamma^c_i=\frac{1}{2}\epsilon^{cab}\tfrac{a_{b}}{a_{a}}\Lambda^I_aR^b_Jf^J_{IK}\omega^K_i-\frac{1}{4}\epsilon^{cab}\tfrac{1}{a_{a}a_{b}}\eta_{LK}\Lambda^I_a\Lambda^J_b f^L_{JI}\omega^K_i.
\end{equation}
The connection $A^a_i$ is now expressed in terms of scale factor, rotation matrices and structure constants
\begin{align}
\label{connBianchi}
\nonumber
    &A^a_i=\Big(\frac{1}{2}\epsilon^{abc}\tfrac{a_{c}}{a_{b}}\Lambda^J_bR^c_Kf^K_{JI}-\frac{1}{4}\epsilon^{abc}\tfrac{1}{a_{b}a_{c}}\eta_{IJ}\Lambda^K_b\Lambda^L_c f^J_{LK}\\
    &+\frac{\gamma}{2N}a_{(a)}R^a_L\big(\eta^{LJ}\dot{\eta}_{JI}+N^A\eta^{LK}\eta_{IJ}f^J_{AK}+N^Af^L_{AI}\big)\Big)\omega^I_i
\end{align}
As the Lagrangian, the expression in terms of elementary functions is not easy to read. However, this calculation shows explicitly that the connection is linearly dependent on the left-invariant one-form, as one expects from (\ref{MaBoVa}).\\

Finally, one wants to find the dependence of $p^I_a$ and $\phi^a_I$ on the scale factors and Euler angles. Since the dependence of $A^a_i$ on the left-invariant one-form is made explicit in Eq.(\ref{connBianchi}), the components of the linear morphism $\phi$ can be easily written
\begin{align}
\nonumber
    \phi^a_I&=\frac{1}{2}\epsilon^{abc}\tfrac{a_{c}}{a_{b}}\Lambda^J_bR^c_Kf^K_{JI}-\frac{1}{4}\epsilon^{abc}\tfrac{1}{a_{b}a_{c}}\eta_{IJ}\Lambda^K_b\Lambda^L_c f^J_{LK}\\
    &+\frac{\gamma}{2N}a_{(a)}R^a_L\big(\eta^{LJ}\dot{\eta}_{JI}+N^A\eta^{LK}\eta_{IJ}f^J_{AK}+N^Af^L_{AI}\big).
\end{align}
While $p^I_a$ can be computed from the metric
$$h_{ij}=\Gamma_{ab}R^a_IR^b_J\omega^I_i\omega^J_j=|\mathrm{det}(p^I_a)|p^c_Ip^c_J\omega^I_i\omega^J_j$$
from which, one obtains
\begin{equation}
    \sqrt{h}=\bar e=|abc||\mathrm{det}(\omega^I_i)|=\sqrt{|\mathrm{det}(p^I_a)|}|\mathrm{det}(\omega^I_i)|
\end{equation}
Recalling that the formula for the densitized dreibein is
$$
    E^i_a=\sqrt{h}e^i_a=|\mathrm{det}(\omega^I_i)|\frac{|a_1a_2a_3|}{a_{(a)}}\Lambda^I_a\xi^i_I
$$
Hence, $p^I_a$ reads
\begin{equation}
\label{momenta}
    p^I_a=\mathrm{sgn}(a_{(a)})|a_ba_c|\Lambda^I_a\ \ \ \ \textrm{with $\epsilon_{abc}=1$.}
\end{equation}\\

The classical theory presents eight disconnected cases, one for each choice of signs for $a,b,c$, in fact, the eigenvalues of the metric cannot vanish. Furthermore, considering all three eigenvalues different from each other at any $t$, if they can be ordered in the whole interval then $\eta_{IJ}$ can be decomposed as in Eq.(\ref{mydecom}). Thus, without loss of generality, one can consider $a>b>c>0\ \forall\ t$, in the vierbein representation this choice means that the densitized dreibein has the same orientation as the left-invariant vectors.

\section{The Bianchi I case \label{sec:Bianchi1}}
In the non-diagonal Bianchi models, the Lagrangian and the connections have a long and difficult expression, so the study with respect to the scale factors and Euler angles is too complicated. The Bianchi I model provides a huge simplification of the formulae. In such a model many terms in the Lagrangian and the connection vanish and they come out very simplified.\\
Considering the Bianchi I model with a metric as in (\ref{mydecom}), the ADM Lagrangian is derived by (\ref{non-diagL}) and it reads
\begin{align}
\nonumber
    \mathcal{L}=&\frac{1}{4N}\sqrt{\mathrm{det}(\Gamma_{AB})}\bigg(\Gamma^{AC}\Gamma^{BD}\dot{\Gamma}_{AB}\dot{\Gamma}_{CD}\\ \nonumber
    &+2\Gamma^{AB}\Gamma_{CD}(R\dot\Lambda)^D_A(R\dot\Lambda)^C_B\\ 
    &+2(R\dot\Lambda)^B_C(R\dot\Lambda)^C_B-\Gamma^{IJ}\dot{\Gamma}_{IJ}\Gamma^{KL}\dot{\Gamma}_{KL}\bigg).
\end{align}
Notice that the Bianchi I hypothesis coincides with a gauge fixing in the ADM formalism: on a generic Bianchi model, imposing the vanishing of shift vector, results in the same Lagrangian. As seen before, the structure constant is always coupled with the shift vector, then, $f^K_{IJ}=0$ and $N^i=0$ remove the same terms. This is not valid for the connection, in which some terms depend on the structure constants only, which emerge from the spin connection.\\
Since the spin connection $\Gamma^a_i$ vanishes in the Bianchi I model, the connection has a really simple expression 
\begin{equation}
\label{conn1}
    \phi^a_I=\frac{\gamma}{2N}a_{(a)}R^a_L\eta^{LJ}\dot{\eta}_{JI}.
\end{equation}
However, its expression as a matrix function of scale factors and Euler angles remains difficult to read. Nevertheless, we can do some manipulation, considering the decomposition of $\eta$, we can rewrite the connection as
\begin{equation}
\label{sconn1}
    \phi^a_I=\frac{\gamma}{2Na_{(a)}}\Lambda^J_a\,\dot{\eta}_{JI}.
\end{equation}
Moreover, the diagonal case can be obtained considering $R^a_I=\delta^a_I$. From the previous expression, we get
$$\phi^a_I=\frac{\gamma}{N}\dot a_{(a)}\delta^a_I,$$
from which, in the isotropic case, given by $a_1=a_2=a_3=a$, the connection reads
$$\phi^a_I=\frac{\gamma}{N}\dot a\delta^a_I,$$
that is the same in \cite{ashtekar2011loop}.\\

Unlike the connection, we can write the Lagrangian explicitly in terms of the metric variables 
\begin{widetext}
\begin{align}
\nonumber
\label{Lag1}
    \mathcal{L}=&\frac{1}{2 N a b c}  \Big(-4 a b c  \left(a \dot b \dot c +b \dot a \dot c +c \dot a \dot b\right)-2 a^2 b^2 c^2 \left(2 \dot{\theta}\dot{\varphi} \cos \psi +\dot{\theta}^2+\dot{\psi}^2+\dot{\varphi}^2\right)\\ \nonumber
    &+a^4 \left(b^2 (\cos \theta\, \dot{\psi}+\sin \theta \sin \psi \,\dot{\varphi})^2+c^2 (\dot{\theta}+\cos \psi \,\dot{\varphi})^2\right)+b^4\left(a^2 (\cos \theta \sin \psi\, \dot{\varphi}-\sin \theta \,\dot{\psi})^2+c^2(\dot{\theta}+\cos \psi\, \dot{\varphi})^2\right)\\
    &+c^4\left(a^2(\cos \theta  \sin \psi\, \dot{\varphi}-\sin \theta\, \dot{\psi})^2+b^2(\cos \theta\,\dot{\psi}+\sin \theta \sin \psi \,\dot{\varphi})^2\right)\Big).
\end{align}
\end{widetext}
The reduction to the diagonal case can be implemented considering the Euler angles constant and null, thus, only the first term remains
$$\mathcal{L}^{\mathrm{diag}}=-\frac{2}{N} \left(a \dot b \dot c +b \dot a \dot c +c \dot a \dot b\right),$$
hence, for the isotropic case, the Lagrangian is proportional to the scalar curvature in the flat FLRW model, at least of a total derivative
$$\mathcal{L}^{\mathrm{iso}}=-\frac{6}{N} a \dot a^2.$$

From the Lagrangian $\mathcal{L}$ we can compute the momenta conjugate to the metric variables. Clearly, we have the momenta
\begin{align*}
    \nonumber
    &\frac{\partial\mathcal{L}_{ADM}}{\partial\dot a}=-\frac{2}{N}(c \dot b+b \dot c),\\
    &\frac{\partial\mathcal{L}_{ADM}}{\partial\dot b}=-\frac{2}{N}(c \dot a+a \dot c),\\
    &\frac{\partial\mathcal{L}_{ADM}}{\partial\dot c}=-\frac{2}{N} (b \dot a+a \dot b).
\end{align*}
The conjugate momenta to the scale factors are the usual ones of the diagonal case. Instead, the conjugate momenta to the angles have a slightly difficult expression. 
\begin{align*}
    &\frac{\partial\mathcal{L}_{ADM}}{\partial\dot{\theta}}=\frac{c}{N a b}(a^2-b^2)^2 (\dot{\theta}+\cos\psi \,\dot{\varphi}),\\ 
    &\frac{\partial\mathcal{L}_{ADM}}{\partial\dot{\psi}}=\frac{1}{N a b c}\bigg(a^4b^2\big(\cos^2\theta\,\dot{\psi}+\sin\theta\cos\theta\sin\psi\,\dot{\phi}\big)\\
    &\ \ \ \ \ \ \ \ \ \ \ \ \ \ \ \ \ \ \ \ \ +b^4a^2\big(\sin^2\theta\,\dot{\psi}-\sin\theta\cos\theta\sin\psi\,\dot{\phi}\big)\\
    &\ \ \ \ \ \ \ \ \ \ \ \ \ \ \ \ \ \ \ \ \ +c^4\Big(a^2\big(\sin^2\theta\,\dot{\psi}-\sin\theta\cos\theta\sin\psi\,\dot{\phi}\big)\\
    &\ \ \ \ \ \ \ \ \ \ \ \ \ \ \ \ \ \ \ \ \ \ \ \ +b^2\big(\cos^2\theta\,\dot{\psi}+\sin\theta\cos\theta\sin\psi\,\dot{\phi}\big)\Big)+\\
    &\ \ \ \ \ \ \ \ \ \ \ \ \ \ \ \ \ \ \ \ \ -2a^2b^2c^2\,\dot{\psi}\bigg),\\
    &\frac{\partial\mathcal{L}_{ADM}}{\partial\dot{\varphi}}=\frac{1}{N a b c}  \bigg(a^4 \Big(b^2 \sin\theta \sin\psi (\cos\theta \,\dot{\psi}+\sin\theta \sin\psi \,\dot{\varphi})\\
    &\ \ \ \ \ \ \ \ \ \ \ \ \ \ \ \ \ \ \ \ \ \ \ \ \ \ \ \ \ \ \ +c^2 \cos\psi (\dot{\theta}+\cos\psi \,\dot{\varphi})\Big)\\ 
    &\ \ \ \ \ \ \ \ \ \ \ \ \ \ \ \ \ \ +b^4\Big( a^2\cos\theta \sin\psi (\cos\theta \sin\psi \,\dot{\varphi}-\sin\theta \,\dot{\psi})\\
    &\ \ \ \ \ \ \ \ \ \ \ \ \ \ \ \ \ \ \ \ \ \ \ +c^2 \cos\psi (\dot{\theta}+\cos\psi \,\dot{\varphi})\Big)\\ 
    &\ \ \ \ \ \ \ \ \ \ \ \ \ \ \ \ \ \ +c^4\Big(a^2 \cos\theta \sin\psi (\cos\theta \sin\psi \,\dot{\varphi}-\sin\theta \,\dot{\psi})\\ 
    &\ \ \ \ \ \ \ \ \ \ \ \ \ \ \ \ \ \ \ \ \ \ \ +b^2 \sin\theta \sin\psi (\cos\theta \,\dot{\psi}+\sin\theta \sin\psi \,\dot{\varphi})\Big)+\\
    &\ \ \ \ \ \ \ \ \ \ \ \ \ \ \ \ \ \ \ -2 a^2 b^2 c^2 (\dot{\theta} \cos\psi+\,\dot{\varphi})\bigg).
\end{align*}

\subsection{Contraints and their algebra for the Bianchi I model \label{sec:const1}}
The Bianchi I model simplifies the expression of the connection with respect to the metric variables. Furthermore, the constraints referred to the Ashtekar's variables, presented in \cite{bojowald2000loopI} for a homogeneous model, read more simply
\begin{align}
\nonumber
    &G_a=\epsilon_{ab}^{\ \ c}\phi^b_I p^I_c,\\
    &\mathcal{D}_I=G_b\phi^b_I,\\ \nonumber
    &\mathcal{S}=-\frac{1}{\gamma^2|\mathrm{det}(p^K_c)|}\Big(p^I_a\phi^a_I p^J_b\phi^b_J-p^I_a\phi^a_J p^J_b\phi^b_I\Big),
\end{align}
where $G_a$ and $\mathcal{S}$ are the Gauss and the scalar constraint respectively. Notice that the Diffeomorphism constraint $\mathcal{D}_I$ does not play any role. It is a linear combination of the Gauss constraint, then it weakly vanishes $\mathcal{D}_I\approx0$. In general, the Diffeomorphism constraint is not relevant for homogeneous models \cite{bojowald2013mathematical}.\\
The algebra of the constraints is easy to find in the phase space $(\phi^I_a,p^b_J)$. Despite $\mathcal{D}_I$ being a linear combination of the other constraints, it can give a non-trivial contribution to the algebra of the constraints. It is convenient to introduce a new scalar constrain $\mathcal{S}'=\gamma^2|\mathrm{det}(p^K_c)|\mathcal{S}$, it is well define because $|\mathrm{det}(p^K_c)|$ is always strictly greater than zero.\\
Furthermore, the Poisson brackets of the constraint with the phase space variables read
\begin{align*}
    &\{G_a,p^I_b\}=k\gamma'\,\epsilon_{ab}^{\ \ c}p^I_c\\
    &\{G_a,\phi^b_I\}=k\gamma'\,\epsilon_{abc}\phi^c_I\\
    &\{\mathcal{D}_I,p^J_a\}=k\gamma'\,\epsilon_{ba}^{\ \ c}\phi^b_Ip^J_c+k\gamma'G_a\delta^J_I\\
    &\{\mathcal{D}_I,\phi^a_J\}=k\gamma'\,\epsilon_{bac}\phi^b_I\phi^c_J\\
    &\{\mathcal{S}',p^I_a\}=-2k\gamma'\Big(p^I_a p^J_b\phi^b_J-p^J_a p^I_b\phi^b_J\Big)\\
    &\{\mathcal{S}',\phi^a_I\}=2k\gamma'\left(\phi^a_I p^J_b\phi^b_J-\phi^a_J p^J_b\phi^b_I\right)
\end{align*}
where $k=\frac{8\pi G}{c^3}$ is the gravitational constant and $\gamma'=\gamma/V_0$. One can compute the Poisson brackets between the constraints. For the Gauss constraint, we get
\begin{align*}
    \{G_a,G_b\}&=k\gamma'\,\epsilon_{bc}^{\ \ d}\left(\epsilon_{ace}\phi^e_Ip^I_d+\epsilon_{ad}^{\ \ e}p^I_e\phi^c_I\right)\\
    &=k\gamma'\left(\delta^c_b\delta_{ad}-\delta^c_a\delta_{db}\right)\phi^d_Ip^I_c=k\gamma'\,\epsilon^c_{de}\epsilon^e_{\ ba}\phi^d_Ip^I_c,
\end{align*}
while for the scalar constraint, we obtain
\begin{align*}
    \{G_a,\mathcal{S}'\}=&k\gamma'\,\epsilon_{ab}^{\ \ c}p^I_c\left(2\phi^b_I p^J_d\phi^d_J-2\phi^b_J p^J_d\phi^d_I\right)+\\
    &-k\gamma'\,\epsilon_{ab}^{\ \ c}\phi^b_I\left(2p^I_c p^J_d\phi^d_J-2p^J_c p^I_d\phi^d_J\right)\\
    =&2k\gamma'G_ap^J_d\phi^d_J-2k\gamma'\,\epsilon_{ab}^{\ \ c}p^I_c\phi^b_J p^J_d\phi^d_I+\\
    &-2k\gamma'G_ap^J_d\phi^d_J+2k\gamma'\,\epsilon_{ab}^{\ \ c}\phi^b_Ip^J_c p^I_d\phi^d_J=0.
\end{align*}
The Poisson brackets with $\mathcal{D}_I$ are proportional to the other constraints, then, at least they weakly vanish.\\
Hence, the algebra generated from the Gauss constraint and the scalar constraint reads
\begin{align}
\nonumber
    &\{G_a,G_b\}=k\gamma'\,\epsilon_{abc}G_c\approx0\\ \nonumber
    &\{G_a,\mathcal{D}_I\}=0\\
    &\{G_a,\mathcal{S}'\}=0\\ \nonumber
    &\{\mathcal{D}_I,\mathcal{S}'\}=G_b\{\phi^b_I,\mathcal{S}'\}\approx0.
\end{align}
The algebra is closed and all the constraints are first-class. It is also evident from the Poisson brackets that the Gauss constraint is the generator of the gauge transformation: it rotates the internal index with the $\mathfrak{su}(2)$ structure constant, while it leaves invariant the scalar quantities and the coordinate indices $I$.\\

Using Eq.(\ref{conn1}) e (\ref{momenta}), the constraints can be written in terms of metric variables. For the scalar constraint, we obtain
\begin{widetext}
\begin{align}
\nonumber
    \mathcal{S}^{\mathrm{ph}}&=\frac{1}{2 N^2 a b c}\Big(-4 a b c  \left(a \dot b \dot c +b \dot a \dot c +c \dot a \dot b\right)-2 a^2 b^2 c^2 \left(2 \dot{\theta}\dot{\varphi} \cos \psi +\dot{\theta}^2+\dot{\psi}^2+\dot{\varphi}^2\right)\\ \nonumber
    &+a^4 \left(b^2 (\cos \theta\, \dot{\psi}+\sin \theta \sin \psi \,\dot{\varphi})^2+c^2 (\dot{\theta}+\cos \psi \,\dot{\varphi})^2\right)+b^4\left(a^2 (\cos \theta \sin \psi\, \dot{\varphi}-\sin \theta \,\dot{\psi})^2+c^2(\dot{\theta}+\cos \psi\, \dot{\varphi})^2\right)\\
    &+c^4\left(a^2(\cos \theta  \sin \psi\, \dot{\varphi}-\sin \theta\, \dot{\psi})^2+b^2(\cos \theta\,\dot{\psi}+\sin \theta \sin \psi \,\dot{\varphi})^2\right)\Big),
\end{align}
\end{widetext}
while the Gauss constraint vanishes identically
\begin{equation}
    G_a^{\mathrm{ph}}=0.
\end{equation}
The vanishing of the Gauss constraint implies, due to the Theorem $I.2.1$ in \cite{Thiemann:2002nj}, that the other constraints are exactly the ADM ones. One can check that $\mathcal{S}^{\mathrm{ph}}$ is the super-Hamiltonian $\mathcal{H}$ obtains from the ADM Lagrangian in (\ref{Lag1}).\\
The vanishing of the Gauss constraint is due to the set of variables chosen, the set $\{a,b,c,\theta,\psi,\varphi\}$ is a set of metric variables, while the Gauss constraint is associated with a gauge transformation. Thus, without the gauge freedom, as in this case, the Gauss constraint vanishes identically.

\section{Quantization of the non-diagonal Bianchi I model \label{sec:quan1}}
The non-diagonal Bianchi I model is strongly linked to the diagonal one. In fact, we will show that the geometrical information is contained in some ``diagonal" quantities and the quantization of the non-diagonal model follows the one proposed in \cite{ashtekar2009loop} for diagonal Bianchi I models. These considerations suggest that the angles do not play any role in the kinematical quantization. We can find a justification in the classical description: a constant rotation of the left-invariant vectors does not affect the description; in fact, a linear combination of left-invariant vectors is a left-invariant vector and it is always possible to lead back to a diagonal metric, i.e. the rotation can be absorbed in $\omega^I$. Considering a time dependant rotation, the linear combination remains in the Lie algebra of the homogeneous space and for each fixed time the previous consideration holds, hence, in each surface at a fixed time it is possible to have a diagonal metric. At the quantum level, it means that the angles do not affect the kinematics and they contribute only to the dynamics of the theory. Thus, it is possible to use the kinematical Hilbert space of the diagonal theory as the core of the Hilbert space that describes the kinematics of the non-diagonal theory.

\subsection{Holonomy operator \label{sec:holo}}
The holonomy is the fundamental operator of the approach with loops. In LQC, the peculiar symmetries of the holonomy due to a simpler phase space are useful to quantize the Hamiltonian and they enable to define of a suitable kinematical Hilbert space that renounces to the $SU(2)$ construction to adopt an abelian one. We want to verify if the properties of the pointwise homogeneous holonomy, studied in the diagonal case, hold in the non-diagonal models.\\
First of all, we want to find the usual expression of the holonomy in terms of trigonometric functions as in Eq.(3.1) in \cite{ashtekar2009loop}. Considering the holonomy along an edge $e_I$ with length $l_I$
\begin{equation}
    \exp\left(\int_{e_{I}} ds\,\dot c^i A^a_i\tau_a\right)=\exp\left(l_I\phi^a_I\tau_a\right).
\end{equation}
In order to do the expansion, the Taylor series of the exponential is considered
\begin{align}
\nonumber
    \exp\left(l_I\phi^a_I\tau_a\right)&=\sum_{n=0}^{\infty}\frac{1}{n!}(l_I\phi^a_I\tau_a)^n\\ \nonumber
    &=\sum_{n=0}^{\infty}\frac{1}{(2n)!}(l_I\phi^a_I\tau_a)^{2n}\\
    &\ \ \ \ +\sum_{n=0}^{\infty}\frac{1}{(2n+1)!}(l_I\phi^a_I\tau_a)^{2n+1}.
\end{align}
To proceed we need to evaluate the square of the argument
\begin{align*}
    (l_I\phi^a_I\tau_a)^2&=l_I^2\phi^a_I\tau_a\phi^b_I\tau_b=l_I^2\phi^a_I\phi^b_I\left(-\tfrac{1}{4}\delta_{ab}+\tfrac{1}{2}\epsilon_{abc}\tau_c\right)\\
    &=-\frac{l_I^2}{4}\frac{\gamma^2}{4N^2}\frac{\delta_{ab}}{c_a c_b}\Lambda^K_a\Lambda^L_b\dot{\eta}_{KI}\dot{\eta}_{LI}\\
    &=-\frac{l_I^2}{4}\frac{\gamma^2}{4N^2}\eta^{KL}\dot{\eta}_{KI}\dot{\eta}_{LI}.
\end{align*}
It is useful to define $c_{II}=\frac{\gamma}{2N}\sqrt{\eta^{KL}\dot{\eta}_{KI}\dot{\eta}_{LI}}$. In such a way the argument of the series reads
\begin{align*}
    &(l_I\phi^a_I\tau_a)^{2n}=\left(-\tfrac{1}{4}l_I^2c_{II}^2\right)^n=(-1)^n\left(\tfrac{1}{2}l_Ic_{II}\right)^{2n}\\
    &(l_I\phi^a_I\tau_a)^{2n+1}=(-1)^n\left(\tfrac{1}{2}l_Ic_{II}\right)^{2n}(l_I\phi^a_I\tau_a)\\
    &\ \ \ \ \ \ \ \ \ \ \ \ \ \ \ \ =(-1)^n\left(\tfrac{1}{2}l_Ic_{II}\right)^{2n+1}\frac{(2\phi^a_I\tau_a)}{c_{II}}.
\end{align*}
Hence, the Taylor series reads
\begin{align}
\label{nondexp}
\nonumber
     \exp\left(l_I\phi^a_I\tau_a\right)&=\sum_{n=0}^{\infty}\frac{1}{(2n)!}(l_I\phi^a_I\tau_a)^{2n}\\ \nonumber
     &\ \ \ \ +\sum_{n=0}^{\infty}\frac{1}{(2n+1)!}(l_I\phi^a_I\tau_a)^{2n+1}\\ \nonumber
     &=\sum_{n=0}^{\infty}\frac{(-1)^n}{(2n)!}\left(\tfrac{1}{2}l_Ic_{II}\right)^{2n}\\ \nonumber
     &\ \ \ \ +\sum_{n=0}^{\infty}\frac{(-1)^n}{(2n+1)!}\left(\tfrac{1}{2}l_Ic_{II}\right)^{2n+1}\frac{(2\phi^a_I\tau_a)}{c_{II}}\\ 
     &=\cos\left(\tfrac{1}{2}l_Ic_{II}\right)+\frac{(2\phi^a_I\tau_a)}{c_{II}}\sin\left(\tfrac{1}{2}l_Ic_{II}\right).
\end{align}
It is evident that the issue is the factor of the sine, which is not present in the isotropic and diagonal cases. However, we can check easily that, if $\phi^a_I=c_{(a)}\delta^a_I$, then $c_{II}=c_I$, where $c_I$ are the diagonal variables (cf. \ref{sec:loop1}). Thus, we find the usual expression $\exp\left(l_I\phi^a_I\tau_a\right)=\cos\left(\tfrac{1}{2}l_Ic_I\right)+2\tau_I\sin\left(\tfrac{1}{2}l_Ic_I\right)$ presented in \cite{ashtekar2009loop}.\\

With this expansion, we can derive some properties of the pointwise holonomy in the non-diagonal Bianchi I model. First of all, the commutation relations are enounced in the following theorem. 
\begin{mytheo}
\label{theocomm}
Considering $SU(2)$ represented in its fundamental representation, i.e $2\times2$ unitary matrices associated to spin $j=\frac{1}{2}$, the usual matrix commutator reads
\begin{align*}
    [&\exp\left(l_I\phi^a_I\tau_a\right),\exp\left(l_J\phi^a_J\tau_a\right)]_{SU(2)\subset M(2,\mathbb{C})}=\\
    &=\frac{4}{c_{II}c_{JJ}}F^c_{IJ}\tau_c\sin\left(\tfrac{1}{2}l_Ic_{II}\right)\sin\left(\tfrac{1}{2}l_Jc_{JJ}\right)
\end{align*}
\end{mytheo}
\begin{proof}
Fixed the representation $j=\frac{1}{2}$ means that the generators of the algebra are proportional to the Pauli matrices $\tau_a=-\frac{i}{2}\sigma_a$, then the commutator defined in $M(2,\mathbb{C})$ still holds. The calculation follows from the expansion (\ref{nondexp}), where the terms proportional to the identity are neglected because they do not contribute to the commutator
\begin{align*}
    [&\exp\left(l_I\phi^a_I\tau_a\right),\exp\left(l_J\phi^a_J\tau_a\right)]_{M(2,\mathbb{C})}=\\
    &=\left[\frac{(2\phi^a_I\tau_a)}{c_{II}}\sin\left(\tfrac{1}{2}l_Ic_{II}\right),\frac{(2\phi^b_J\tau_b)}{c_{JJ}}\sin\left(\tfrac{1}{2}l_Jc_{JJ}\right)\right]_{M(2,\mathbb{C})}\\
    &=\frac{(2\phi^a_I)}{c_{II}}\sin\left(\tfrac{1}{2}l_Ic_{II}\right)\frac{(2\phi^b_J)}{c_{JJ}}\sin\left(\tfrac{1}{2}l_Jc_{JJ}\right)\left[\tau_a,\tau_b\right]_{\mathfrak{su}(2)}\\
    &=\frac{4}{c_{II}c_{JJ}}\epsilon_{abc}\phi^a_I\phi^b_J\tau_c\sin\left(\tfrac{1}{2}l_Ic_{II}\right)\sin\left(\tfrac{1}{2}l_Jc_{JJ}\right)
\end{align*}
\end{proof}
This is a completely generic formula valid for any homogeneous holonomy in Bianchi I. It is evident that, for $I=J$, the skew-symmetry of $F^c_{IJ}$ assures the vanishing of the commutator of holonomy along the same edge. Moreover, it holds also in the diagonal case in a slightly simplified expression
\begin{align*}
    [&\exp\left(l_Ic_I\tau_I\right),\exp\left(l_Jc_J\tau_J\right)]=\\
    &=4\,\epsilon_{IJk}\tau_k\sin\left(\tfrac{1}{2}l_Ic_{I}\right)\sin\left(\tfrac{1}{2}l_Jc_{J}\right).
\end{align*}
Instead, the isotropic case is a different one. Due to the isotropy, one is unable to distinguish different directions of the space, then the holonomy along each edge must commute each other and the commutator always vanishes. Thus, the abelian approach in the isotropic case finds its reason in the group of isotropic holonomies that is abelian, further the identical null Gauss constraint.\\

In an analogous way, it is possible to derive another formula for the commutation property. It was originally derived by M.Bojowald in \cite{bojowald2000loopII} for the diagonal case, but it is possible to generalize it to any homogeneous connection and to prove that it holds for $SU(2)$.
\begin{mylemma}
\label{lemma1}
Let $g,h\in SU(2)$ in the fundamental representation. Then
\begin{equation}
    gh=hg+h^{-1}g+hg^{-1}-tr(hg^{-1})
\end{equation}
\end{mylemma}
\begin{proof}
There are two ways to prove the lemma: using (\ref{nondexp}) or with the quaternions. To prove it for the whole $SU(2)$ the quaternions give a more general formalism.\\
Considering the isomorphism between $SU(2)$ and the unit quaternions 
\begin{align*}
    F:&SU(2)\to\mathbb{H}\\
    &\begin{pmatrix}
      \alpha & \beta\\
      -\bar{\beta} & \bar{\alpha}
    \end{pmatrix}\mapsto q=\alpha+\beta j\ \ \ \ \ \alpha,\beta\in\mathbb{C}
\end{align*}
Recalling the quaternions algebra $ij=k,jk=i,ki=j$ and $i^2=j^2=k^2=ijk=-1$. It is possible to define the complex conjugation as the change of the sign of $i,j,k$ and for unit quaternions $q^{-1}=\bar{q}$. Notice that, using the fundamental representation of $SU(2)$, there is a link between trace and real part: $tr(g)=2\Re(F(g))=2\Re(\alpha)$.\\
Thus, it is enough to prove the formula for the unit quaternions, so the problem is reduced to doing some simple algebra. Let $q_1=a+bi+cj+dk$ and $q_2=w+xi+yj+zk$, the single terms are computed
\begin{align*}
    &q_1q_2-q_2q_1=2(cz-dy)i+2(dx-bz)j+2(by+cx)k\\
    &q_2^{-1}q_1=(aw+bx+cy+dz)+(bw-ax+cz-dy)i\\
    &\ \ \ \ \ \ \ \ \ \ +(cw-ay+dx-bz)j+(dw-az+by-cx)k\\
    &q_2q_1^{-1}=(aw+bx+cy+dz)+(-bw+ax+cz-dy)i\\
    &\ \ \ \ \ \ \ \ \ \ +(ay-cw+dx-bz)j+(az-dw+by-cx)k\\
\end{align*}
Where the real parts are evident. From which
\begin{align*}
    q_1&q_2-q_2q_1-q_2^{-1}q_1-q_2q_1^{-1}=\\
    &=-2(aw+bx+cy+dz)=-2\Re(q_2q_1^{-1}).
\end{align*}
Where $\Re$ indicates the real part.
\end{proof}
Using the expansion (\ref{nondexp}), after some tedious calculations, we find that the formula by M.Bojowald is a subcase of this Lemma. This can be proved by the comparison between the formulas. There is no evident symmetry in Bojowald's hypothesis for which $tr(hg)=tr(hg^{-1})$, but it is easy to check this equivalence from the expansion in the diagonal case
\begin{align*}
     &tr(hg^{\pm1})=tr(e^{a\tau_I}e^{\pm b\tau_J})=\\
     &=2\cos(a/2)\cos(b/2)\mp2\delta_{IJ}\sin(a/2)\sin(b/2)\ \ a,b\in\mathbb{R}.
\end{align*}
Since $I\neq J$ is required in Bojowald's Lemma, the sign does not contribute.\\

The curvature operator follows the usual definition of the canonical LQC but the choice of the plaquette must be different. Recalling that
\begin{equation}
    F_{ij}^a=2\lim_{Ar_{\Box_{IJ}}\to0}\left(\frac{h_{\Box_{IJ}}-1}{Ar_{\Box_{IJ}}}\tau_a\right)\omega^I_i\omega^J_j
\end{equation}
where $\Box_{IJ}$ is the plaquette, i.e. a rectangular closed path with the edges along $\xi_I$ and $\xi_J$. The limit does not exist, so we need to consider the plaquette with the minimal area, as in \ref{sec:loop1}. We want to emulate the procedure for the diagonal model and find a natural choice of the plaquette from which to compute the curvature.\\
In the non-diagonal case, $p^I_a$ is not diagonal but with a correct choice of the plaquette, we can bypass the problem and find three fluxes that are the same as the diagonal case. Considering the vectors
\begin{equation}
    \zeta_a=\Lambda^I_a\xi_I,
\end{equation}
they are elements of the Lie algebra of $\Sigma$ and it is easy to check that $[\zeta_a,\zeta_b]=0$. Moreover, along these vectors, the densitized dreinbein vectors are diagonal
\begin{equation}
     E^i_a=|\mathrm{det}(\omega^I_i)|\frac{|a_1a_2a_3|}{a_{(a)}}\Lambda^I_a\xi^i_I=|\mathrm{det}(\omega^I_i)|p_{(a)}\zeta^i_a,
\end{equation}
where the expression of $p_a$ can be derived by Eq.(\ref{momenta})
\begin{equation}
    p_a=\mathrm{sgn}(a_{a})|a_ba_c|\ \ \ \ \textrm{with $\epsilon_{abc}=1$.}
\end{equation}
Hence, since the choice of the plaquette for the regularization of the curvature is completely arbitrary, we choose the plaquette $\Box_{ab}$ as a rectangular closed path with edges along $\zeta_a$ and $\zeta_b$, that lies in the plane $I$-$J$ rotated.\\
The pointwise holonomy along the edge $\zeta_a$ can be easily computed
\begin{align}
    \exp\left(\int_{\zeta_a}ds\,\dot c^iA^a_i\tau_a\right)=\exp\left(l_a\zeta_a^i\phi^b_I\omega^I_i\tau_b\right)=\exp\left(\lambda_a\mathring{\phi}^b_a\tau_b\right)
\end{align}
where $\mathring{\phi}^b_a$ is defined as $\mathring{\phi}^b_a=L_a\Lambda^I_a\phi^b_I$. $L_a$ is the fiducial length of the $a$-th edge of the rotated fiducial cell, which is equal to the canonical fiducial cell $L_a=L_I$ if $\delta_{aI}=1$. In fact, the fiducial metric is invariant
\begin{equation}
    \left(\prescript{0}{}{q}^{-1}\right)^{ij}=\delta^{IJ}\xi^i_I\xi^j_J=\delta^{IJ}R^a_I\zeta_a^iR^b_J\zeta^j_b=\delta^{ab}\zeta^i_a\zeta^j_b.
\end{equation}
Notice that $\mathring{\phi}^b_a$ has an explicit expression in terms of metric variables
\begin{equation}
   \mathring{\phi}^b_a=\frac{\gamma L_{(a)}}{2Na_{(b)}}\Lambda^J_b\Lambda^I_a\,\dot{\eta}_{JI}.
\end{equation}
Writing $\mathring{\phi}^b_a$, an interesting property becomes evident
\begin{equation}
   \mathring{\phi}^b_a= \left(
\begin{array}{ccc}
 L_1c_1 & \frac{\gamma L_2}{2N}\frac{b^2-a^2}{a}\omega_3 & \frac{\gamma L_3}{2N}\frac{a^2-c^2}{a}\omega_2\\
 \frac{\gamma L_1}{2N}\frac{b^2-a^2}{b}\omega_3 & L_2c_2 & \frac{\gamma L_3}{2N}\frac{c^2-b^2}{b}\omega_1\\
 \frac{\gamma L_1}{2N}\frac{a^2-c^2}{c}\omega_2 & \frac{\gamma L_2}{2N}\frac{c^2-b^2}{c}\omega_1 & L_3 c_3
\end{array}
\right)
\end{equation}
where $c_1,c_2,c_3$ are the diagonal components of the connection as in the diagonal case
\begin{equation}
    c_a=\frac{\gamma}{N}\frac{da_a}{dt},
\end{equation}
and $\omega_1,\omega_2,\omega_3$ are defined as
\begin{align*}
    &\omega_1=\cos\theta\sin\psi\,\dot{\varphi}-\sin\theta\,\dot{\psi},\\
    &\omega_2=\sin\theta\sin\psi\,\dot{\varphi}+\cos\theta\,\dot{\psi},\\
    &\omega_3=\dot{\theta}+\cos\psi\,\dot{\varphi}.
\end{align*}
The connection is composed of a diagonal part that is exactly the connection of the diagonal case and out-of-diagonal terms, which depend on angles and linearly in their conjugate momenta. Notice that for a constant rotation the connection is diagonal. Thus, since the holonomies are diagonal, one leads back to the diagonal case for both the kinematical and the dynamical theory.

\subsection{Quantum geometry \label{sec:QuaGeo}}
In the previous section, with a suitable choice of the plaquette, the link with the diagonal case emerges. The role of the diagonal variables in quantum geometry of the non-diagonal Bianchi I model will be examined in depth.\\
With the choice of the plaquette $\Box_{ab}$ as above, one of the densitized dreibein vectors is orthogonal to the plaquette and its diagonal component can represent the area, so the argument for the diagonal case can be emulated also in the non-diagonal one. Moreover, the flux of the electric field over the face $\tilde{\sigma}_a$ of the rotated fiducial cell is proportional to $p_a$.\\ 
Recalling that the fiducial length is invariant under rotation of the edges, let's introduce new variables via rescaling $\tilde{p}_a=L_bL_cp_a$ with $\epsilon_{abc}=1$, to be coherent with the notation in \ref{sec:loop1}. 
\begin{mylemma}
\label{area}
    The area of $\Tilde{\sigma}_a$ is $|\Tilde{p}_a|$
\end{mylemma}
\begin{proof}
    Considering the face $\Tilde{\sigma}_3$ on the rotated fiducial cell. Let $\{x^i\}$ local coordinates in which $\xi^i_a=\delta^i_a$. Since $[\zeta_a,\zeta_b]=0$, then exist local coordinates $\{\Tilde{x}^i\}$ such that $\zeta^i_a=\delta^i_a$.\\
    We can write the metric on $\Sigma$ in these new coordinates
    $h=a_1^2(d\Tilde{x}^1)^2+a_2^2(d\Tilde{x}^2)^2+a_3^2(d\Tilde{x}^3)^2$. Let $\iota:\Tilde{\sigma}_3\xrightarrow{}\Sigma$ the inclusion map and $F:\Sigma\to\Sigma$ a rotation such that $F_*\xi_1=\zeta_1$ and so on for $2$ and $3$.\\
    For $\Tilde{\sigma}_3$ we have:
    \begin{enumerate}
        \item $\Tilde{\sigma}_3=F(\sigma_{3})$, where $\sigma_3$ is the face of the fiducial cell.
        \item The metric on the surface $\Tilde{\sigma}_3$ is $q=\iota^*h=a_1^2(d\Tilde{x}^1)^2+a_2^2(d\Tilde{x}^2)^2$. Hence, the volume form is $\textit{Vol}_q=|a_1a_2|d\Tilde{x}^1\wedge d\Tilde{x}^2=|p_3|d\Tilde{x}^1\wedge d\Tilde{x}^2$. 
    \end{enumerate}
   Thus, the area of $\Tilde{\sigma}_3$ reads
   \begin{align*}
       \textit{Ar}(\Tilde{\sigma}_3)&=\int_{\Tilde{\sigma}_3}|p_3|d\Tilde{x}^1\wedge d\Tilde{x}^2=|p_3|\int_{F(\sigma_3)}d\Tilde{x}^1\wedge d\Tilde{x}^2\\
       &=|p_3|\int_{\sigma_3}F^*(d\Tilde{x}^1\wedge d\Tilde{x}^2)=|p_3|\int_{\sigma_3}\omega^1\wedge\omega^2\\
       &=|p_3|L_1L_2=|\Tilde{p}_3|.
   \end{align*}
   The same procedure can be repeated for the other two faces.
\end{proof}

Considering also rescaled connection variables $\Tilde{c}_a=L_ac_a$, the couple $(\Tilde{c}_a,\Tilde{p}_b)$ is exactly the one of the diagonal case. Hence, we can choose the phase-space $\{\Tilde{p}_1,\Tilde{p}_2,\Tilde{p}_3,\theta,\psi,\varphi,\Tilde{c}_1,\Tilde{c}_2,\Tilde{c}_3,\pi_{\theta},\pi_{\psi},\pi_{\varphi}\}$ with (non-vanishing) Poisson brackets
\begin{align*}
    &\{\Tilde{c}^a,\Tilde{p}_b\}=k\gamma'\delta^a_b,\\
    &\{\theta,\pi_{\theta}\}=1,\\
    &\{\psi,\pi_{\psi}\}=1,\\
    &\{\varphi,\pi_{\varphi}\}=1.
\end{align*}

Considering the quantum states $|\tilde{p}_1,\tilde{p}_2,\tilde{p}_3,\theta,\psi,\varphi\rangle$, these states are eigestates of quantum geometry. For the Lemma \ref{area}, in such a state the face $\Tilde{\sigma}_3$ in the plane ${\zeta_1,\zeta_2}$ of the rotated fiducial cell has area $|\tilde{p}_3|$. Furthermore, the classical volume has the same expression as the diagonal one, just like the volume operator. We obtain the classical expression
\begin{align}
\nonumber
    \textit{Vol}=&\int d^3x\sqrt{\left|\tfrac{1}{6}\epsilon^{abc}\epsilon_{ijk}E^i_aE^j_bE^k_c\right|}\\ \nonumber
    &=V_0\sqrt{\left|\tfrac{1}{6}\epsilon^{abc}\epsilon_{IJK}p^I_ap^J_bp^K_c\right|}\\ \nonumber
    &=V_0\sqrt{\left|\tfrac{1}{6}\epsilon^{abc}\epsilon_{IJK}\Lambda^I_a\Lambda^J_b\Lambda^K_cp_ap_bp_c\right|}\\
    &=V_0\sqrt{|p_1p_2p_3|}=\sqrt{|\tilde{p}_1\tilde{p}_2\tilde{p}_3|}
\end{align}
Thus, all the geometry information is contained in the ``diagonal" variables $\Tilde{p}_a$, which represent the fluxes, with their conjugate momenta $\Tilde{c}_a$. Hence, as in the diagonal case, the couple $(\Tilde{p}_a,\Tilde{c}_b)$ can be quantized in the loop formalism and their action on the states is
\begin{align}
    &\hat{\Tilde{p}}_1|\Tilde{p}_1,\Tilde{p}_2,\Tilde{p}_3,\theta,\psi,\varphi\rangle=\Tilde{p}_1|\Tilde{p}_1,\Tilde{p}_2,\Tilde{p}_3,\theta,\psi,\varphi\rangle,\\
    &\widehat{\exp(i\lambda \Tilde{c_1})}|\Tilde{p}_1,\Tilde{p}_2,\Tilde{p}_3,\theta,\psi,\varphi\rangle=|\Tilde{p}_1-k\gamma\hbar\lambda,\Tilde{p}_2,\Tilde{p}_3,\theta,\psi,\varphi\rangle,
\end{align}
and similarly for $\hat{\Tilde{p}}_2$ and $\hat{\Tilde{p}}_3$.\\

Since $\hat{\tilde{p}}_a$ is well-defined as an operator, the quantization of the volume is trivial. On the previously defined states, the volume operator acts as
\begin{equation*}
\hat{V}|\Tilde{p}_1,\Tilde{p}_2,\Tilde{p}_3,\theta,\psi,\varphi\rangle=\sqrt{|\tilde{p}_1\tilde{p}_2\tilde{p}_3|}|\Tilde{p}_1,\Tilde{p}_2,\Tilde{p}_3,\theta,\psi,\varphi\rangle
\end{equation*}
The information about quantum geometry is wholly encoded by the fluxes $\Tilde{p}_a$. The angles $\theta,\psi,\varphi$ represent how much non-diagonal is the state, i.e. greater the values more the dreibein is rotated with respect to the left-invariant vectors, which is important in the choice of a suitable plaquette on which compute the curvature operator. As already seen, this choice is arbitrary and we can choose the plaquette whose area is proportional to $\Tilde{p}_a$. Furthermore, without any ambiguity, the volume depends only on momenta. At a quantum level, the Volume operator has the same eigenstates of the fluxes and eigenvalues that depend only on the eigenvalues of the fluxes.\\

The kinematical states of the theory are linear combinations of $|\Tilde{p}_1,\Tilde{p}_2,\Tilde{p}_3,\theta,\psi,\varphi\rangle$, in which $\Tilde{p}_a$ are quantized in the loop framework. About the angles, it is possible to use the Weyl quantization. The fundamental operators read
\begin{align}
    &e^{i\xi\hat{\theta}}|\Tilde{p}_1,\Tilde{p}_2,\Tilde{p}_3,\theta,\psi,\varphi\rangle=e^{i\xi\theta}|\Tilde{p}_1,\Tilde{p}_2,\Tilde{p}_3,\theta,\psi,\varphi\rangle,\\
    &e^{i\eta\hat{\pi}_{\theta}}|\Tilde{p}_1,\Tilde{p}_2,\Tilde{p}_3,\theta,\psi,\varphi\rangle=|\Tilde{p}_1,\Tilde{p}_2,\Tilde{p}_3,\theta+\eta,\psi,\varphi\rangle,
\end{align}
and similarly for $\psi$ and $\varphi$.\\
In this vector space, a norm can be defined such that it is the same as the diagonal case. The diagonal states are the ones with null angles $|\Tilde{p}_1,\Tilde{p}_2,\Tilde{p}_3,0,0,0\rangle$. Thus, is always possible to write any state as a linear transformation of a diagonal one
\begin{align*}
&|\Tilde{p}_1,\Tilde{p}_2,\Tilde{p}_3,\varphi_1,\varphi_2,\varphi_3\rangle=\\
&=e^{i\varphi_1\hat{\pi}_{\theta}}e^{i\varphi_2\hat{\pi}_{\psi}}e^{i\varphi_3\hat{\pi}_{\varphi}}|\Tilde{p}_1,\Tilde{p}_2,\Tilde{p}_3,0,0,0\rangle.
\end{align*}
To define a kinematical Hilbert space, we need to equip it with a scalar product. From the quantization properties, the scalar product must induce a norm on the basis of states that have to satisfy
\begin{align*}
    &|||\Tilde{p}_1,\Tilde{p}_2,\Tilde{p}_3,\varphi_1,\varphi_2,\varphi_3\rangle||^2=\\
    &=\langle\Tilde{p}_1,\Tilde{p}_2,\Tilde{p}_3,\varphi_1,\varphi_2,\varphi_3|\Tilde{p}_1,\Tilde{p}_2,\Tilde{p}_3,\varphi_1,\varphi_2,\varphi_3\rangle\\ 
    &=\langle\Tilde{p}_1,\Tilde{p}_2,\Tilde{p}_3,0,0,0|e^{-i\varphi_3\hat{\pi}_{\varphi}}e^{-i\varphi_2\hat{\pi}_{\psi}}e^{-i\varphi_1\hat{\pi}_{\theta}}\times\\
    &\ \ \ \ \ \ \ \ \ \ \ \ \ \ \ \ \ \ \ \ \ \times e^{i\varphi_1\hat{\pi}_{\theta}}e^{i\varphi_2\hat{\pi}_{\psi}}e^{i\varphi_3\hat{\pi}_{\varphi}}|\Tilde{p}_1,\Tilde{p}_2,\Tilde{p}_3,0,0,0\rangle\\
    &=\langle\Tilde{p}_1,\Tilde{p}_2,\Tilde{p}_3,0,0,0|\Tilde{p}_1,\Tilde{p}_2,\Tilde{p}_3,0,0,0\rangle\\
    &:=|||\Tilde{p}_1,\Tilde{p}_2,\Tilde{p}_3\rangle||^2_{diag}=1,
\end{align*}
in which $||\cdot||_{diag}$ is the norm of the state in the diagonal theory.\\
The property of leading back the scalar product to the diagonal one holds not only for the norm. In fact, consider a generic scalar product between two states
\begin{align*}
    &\langle\Tilde{p}_1,\Tilde{p}_2,\Tilde{p}_3,\varphi_1,\varphi_2,\varphi_3|\Tilde{p}'_1,\Tilde{p}'_2,\Tilde{p}'_3,\varphi'_1,\varphi'_2,\varphi'_3\rangle=\\
    &=\langle\Tilde{p}_1,\Tilde{p}_2,\Tilde{p}_3,0,0,0|e^{i(\varphi_1-\varphi'_1)\hat{p}_{\theta}}e^{i(\varphi_2-\varphi'_2)\hat{p}_{\psi}}\times\\
    &\ \ \ \ \ \ \ \ \ \ \ \ \ \ \ \ \ \ \ \ \ \ \ \ \ \ \ \ \times e^{i(\varphi_3-\varphi'_3)\hat{p}_{\varphi}}|\Tilde{p}'_1,\Tilde{p}'_2,\Tilde{p}'_3,0,0,0\rangle\\
    &=\langle\Tilde{p}_1,\Tilde{p}_2,\Tilde{p}_3,0,0,0|\Tilde{p}'_1,\Tilde{p}'_2,\Tilde{p}'_3,\varphi_1-\varphi'_1,\varphi_2-\varphi'_2,\varphi_3-\varphi'_3\rangle.
\end{align*}
If $\varphi_r=\varphi'_r$, we find the scalar product between diagonal states. So, the scalar product between equally rotated states is imposed to be
\begin{align*}
    &\langle\Tilde{p}_1,\Tilde{p}_2,\Tilde{p}_3,\varphi_1,\varphi_2,\varphi_3|\Tilde{p}'_1,\Tilde{p}'_2,\Tilde{p}'_3,\varphi_1,\varphi_2,\varphi_3\rangle=\\
    &\langle\Tilde{p}_1,\Tilde{p}_2,\Tilde{p}_3|\Tilde{p}'_1,\Tilde{p}'_2,\Tilde{p}'_3\rangle_{diag}=\delta_{\Tilde{p}_1\Tilde{p}'_1}\delta_{\Tilde{p}_2\Tilde{p}'_2}\delta_{\Tilde{p}_3\Tilde{p}'_3}.
\end{align*}
This is coherent with our interpretation of the quantum kinematics of the angles.\\ 
There is no unique choice of a scalar product that induces the properties above. However, we require that $|\Tilde{p}_1,\Tilde{p}_2,\Tilde{p}_3,\varphi_1,\varphi_2,\varphi_3\rangle|$ are orthonormal states then, in according with the quantization procedure, we are allowed to choose as the scalar product 
\begin{align}
    \label{scalprod}
    \nonumber &\langle\Tilde{p}_1,\Tilde{p}_2,\Tilde{p}_3,\varphi_1,\varphi_2,\varphi_3|\Tilde{p}'_1,\Tilde{p}'_2,\Tilde{p}'_3,\varphi'_1,\varphi'_2,\varphi'_3\rangle\\
    &:=\delta_{\Tilde{p}_1\Tilde{p}'_1}\delta_{\Tilde{p}_2\Tilde{p}'_2}\delta_{\Tilde{p}_3\Tilde{p}'_3}\delta(\varphi_1-\varphi'_1)\delta(\varphi_2-\varphi'_2)\delta(\varphi_3-\varphi'_3)
\end{align}
The kinematical Hilbert space of the theory is equipped with the scalar product (\ref{scalprod}) and it has an orthonormal basis $\{|\Tilde{p}_1,\Tilde{p}_2,\Tilde{p}_3,\theta,\psi,\varphi\rangle\}$. In the kinematical Hilbert space, the angles do not play any role, hence, the kinematical theory reflects the properties of the geometric operators.

\section{Representation quantization \label{sec:repp}}
The quantization of homogeneous models was already implemented by M.Bojowald in his work 'Mathematical Structure of Loop Quantum Cosmology: Homogeneous Models' \cite{bojowald2013mathematical}, in which the same procedure of the full theory is used to find the kinematical Hilbert space and the fundamental operators: pointwise holonomy and the momenta operator. The theory is similar to the LQG but it is described in terms of the homogeneous part of the connection $\phi^a_I$ instead of the full one. Nevertheless, in the \ref{sec:const1}, the cosmological theory of the non-diagonal Bianchi I model emerges naturally with a null Gauss constraint. Moreover, the holonomy has the same commutation relation of the diagonal case by Theorem \ref{theocomm}. Hence, we want the theory to be expressed as a $U(1)^3$ theory instead of $SU(2)$, as in the usual approach to diagonal models. To do that, we implement the kinematical quantization procedure shown in \cite{bojowald2013mathematical} in our case, imposing that the states are representations of $U(1)^3$. This formulation seems to be more general than the one proposed in \ref{sec:QuaGeo}, as it holds for any Bianchi models, but it presents some issues and problems with physical interpretation.

\subsection{$U(1)^3$-holonomy \label{U13}}
The choice of the Ashtekar variables $c_1,c_2,c_3$ for the cosmology, seems to be the only reasonable one. However, in the non-diagonal case, there is ambiguity. These variables do not emerge naturally from the geometry but they are useful to show the relation with the diagonal case. We are interested in looking for a more natural set of variables. To do that the expansion in (\ref{nondexp}) can be useful, we can choose the argument of the trigonometric functions as the configuration variables
\begin{equation}
    c_{II}=\sqrt{\sum_a \phi^a_I\phi^a_I}.
\end{equation}
The square root exists because the argument is a sum of squares (in this Section we do not use the Einstein summation convention).\\
Moreover, it is possible to find the conjugate momenta in terms of connection and dreibein. Considering a function $F^J$, imposing the invariance of the Poisson brackets
\begin{align}
    k\gamma'\delta_I^J=\{c_{II},F^J\}_{\phi^a_I,p^I_a}=k\gamma'\frac{1}{c_{II}}\sum_K\delta^K_I\phi^c_K\frac{\partial F^J}{\partial p^K_c},
\end{align}
we obtain $$\frac{\partial F^J}{\partial p^K_c}\propto(\phi^{-1})^J_c,$$ Hence, the momenta read
\begin{equation}
    p^I=c_{II}\sum_a(\phi^{-1})^I_ap^I_a.
\end{equation}
Notice that also $$F^J=\frac{1}{c_{JJ}}\sum_a p^J_a\phi^a_J$$ give us the correct Poisson bracket but the first one will result more useful in the next Section.\\

For the quantisation, the procedure proposed by M.Bojowald in \cite{bojowald2013mathematical} is implemented in the $U(1)^3$ case.\\
In particular, to adapt the formalism to the original paper, rescaled variables are considered
\begin{equation}
    \tilde{c}_{II}=L_I\sqrt{\sum_a \phi^a_I\phi^a_I},\ \ \ \ \ \tilde{p}^I=L_JL_K p^I\ \ \ \textrm{with $\epsilon_{IJK}=1$.}
\end{equation}
The holonomy-flux algebra is implemented considering representations of holonomies as kinetic states of the theory with quantum number $n$ label of representation. Thus, a state reads $\rho_{\lambda,n}(g_I)=(\exp(i\lambda\tilde{c}_{II}))^n$, with $\lambda\in\mathbb{Q}$ and $n\in\mathbb{N}$.\\
The binary operations have really simple expressions, the multiplication reads
\begin{align}
    \nonumber
    \rho_{\lambda_1,n_1}(g_I)\cdot\rho_{\lambda_2,n_2}(g_I):=&\rho_{z,n_1}(g_I)^{N_1}\rho_{z,n_2}(g_I)^{N_2}\\
    &=\rho_{z,N_1n_1+N_2n_2}(g_I)
\end{align}
with $z$ maximal rational such that $\lambda_1=N_1z$ and $\lambda_2=N_2z$. While multiplication between two elements with $I\neq J$ is the tensor product. The star operator is $\rho_{\lambda,n}(g_I)^*=\rho_{\lambda,-n}(g_I)$. For the inner product, the Haar measure on $U(1)$ is required. It can be defined by
\begin{align*}
    \mu_H(S)=\frac{1}{2\pi}m(f^{-1}(S)),
\end{align*}
for each $S\subset U(1)$, where $f$ is the function $f:[0,2\pi]\to U(1),\ t\mapsto (\cos(t),\sin(t))$ and $m$ is the usual Borel measure on the real line. Hence, the measure for the inner product is
\begin{equation}
    \int_{U(1)}d\mu_H(\exp(z\lambda\tilde{c}_{II}))=\frac{1}{2\pi}\int_0^{2\pi/z}d(z\tilde{c}_{II}).
\end{equation}
From this, we can define the inner product as
\begin{align}
\label{abelinn}
\nonumber
    &(\rho_{\lambda_1,n_1}(g_I),\rho_{\lambda_2,n_2}(g_J)):=\\
    &:=\frac{1}{(2\pi)^3}\prod_K\int_0^{2\pi/z}d(z\tilde{c}_{KK})\rho_{\lambda_1,-n_1}(g_I)\cdot\rho_{\lambda_2,n_2}(g_J).
\end{align}
Notice that for $I\neq J$ the inner product vanishes unless $\lambda_1n_1=0=\lambda_2n_2$. If $I=J$ one obtains
\begin{align*}
    &(\rho_{\lambda_1,n_1}(g_I),\rho_{\lambda_2,n_2}(g_I))\\
    &=\frac{1}{2\pi}\int_0^{2\pi/z}d(z\tilde{c}_{II})\,\rho_{\lambda_1,-n_1}(g_I)\cdot\rho_{\lambda_2,n_2}(g_I)\\
    &=\frac{1}{2\pi}\int_0^{2\pi} dx\, \rho_{N_2n_2-N_1n_1}(e^{ix})=\begin{cases}
    0 & \textrm{if}\ \lambda_1n_1\neq\lambda_2n_2\\
    1 & \textrm{if}\ \lambda_1n_1=\lambda_2n_2
    \end{cases}
\end{align*}
All the definitions above are coherent with the isotropic case studied in \cite{bojowald2013mathematical}.\\
Recalling that $\rho_n(\tau_I)=i$ for the $U(1)$ representations, the momentum operator can be derived from the general case,
\begin{equation}
    \hat{\tilde{p}}^I\rho_{\lambda,n}(g_J)=8\pi\gamma \ell_P^2\lambda\delta^I_J\rho_{\lambda,n}(g_J)
\end{equation}
With multiplication operator $\rho_{\lambda,n}(g_I)$ and momentum operator $\hat{\tilde{p}}^I$ one now can compute the commutators acting on a state $\psi(g_I)=\rho_{\lambda_2,n_2}(g_I)$ (for $J\neq I$ is the same since $J$-terms does not contribute), the calculation is pretty simple and one obtains
\begin{equation}
    [\rho_{\lambda,n}(g_I),\hat{\tilde{p}}^I]=-8\pi\gamma\ell_P^2\lambda\rho_{\lambda,n}(g_I).
\end{equation}
There is no reordering operator $\hat{R}^I$, then the commutator is exactly the quantization of the Poisson bracket. It is possible to show that the reordering operator vanishes in the abelian theory
\begin{align*}
    &\hat{R}^I(\rho_{\lambda_1,n_1}(g_I)\cdots\rho_{\lambda_N,n_N}(g_I))\\
    &=iz\sum_{l=1}^N N_l\rho_{\lambda_1,n_1}(g_I)\cdots \rho_{z,n_l}(g_I)^{N_l}\cdots\rho_{\lambda_N,n_N}(g_I)+\\
    &\ \ \ \ \ \ \  -i\sum_{l=1}^N\lambda_l\,\rho_{\lambda_1,n_1}(g_I)\cdots\rho_{\lambda_N,n_N}(g_I)=0
\end{align*}
with $z$ maximal rational such that $\lambda_l=zN_l$. The reason is that refinement does not exist in the abelian theory; in fact, the action of $\rho_{\lambda,0}(g_I)$ as multiplication operator is trivial. In the $SU(2)$ theory, the reordering operator reads $R^I_a=-(8\pi\gamma i\ell_P^2)^{-1}\tilde{p}^I_a(\rho_{z,0}-1)$, but in the abelian theory the term in the parenthesis acts as a null projector on every state, then $\hat{R}^I=0$.\\

This Hilbert space is isomorphic to the LQC's usual one: the space of functions on the Bohr compactification $\overline{\mathbb{R}}_{\textit{Bohr}}$ of the real line (i.e. almost periodic functions). It exists a map between the abelian states and almost periodic functions $B:\rho_{\lambda,n}(g_I)\mapsto\exp(i\lambda n \tilde{c}_{II})$.\\
This map is a $^*$-algebra morphism and commutes with the action of $\hat{\tilde{p}}^I$ \cite{bojowald2013mathematical}. The map is surjective: given any $\exp(ia\tilde{c}_{II})$ is is always possible to find $\rho_{\lambda,n}(g_I)$ such that $\lambda n=a$ and so, to satisfy $B(\rho_{\lambda,n}(g_I))=\exp(ia\tilde{c}_{II})$. $B$ is also injective \footnote{The injectivity of $B$ descents also from the fact that in the original paper \cite{bojowald2013mathematical} $B$ it is an isometry.} because of Lemma \ref{lemma:element}.
\begin{mylemma}
\label{lemma:element}
Two elements $\rho_{\lambda_1,n_1}(g_I)$ and $\rho_{\lambda_2,n_2}(g_I)$ such that $\lambda_1 n_1=\lambda_2 n_2$ are the same point in the Hilbert space with the scalar product defined in (\ref{abelinn}).
\end{mylemma}
\begin{proof}
Let $\rho_{\lambda_1,n_1}(g_I)$ and $\rho_{\lambda_2,n_2}(g_I)$ such that $\lambda_1 n_1=\lambda_2 n_2$. Recalling that $(\rho_{\lambda_1,n_1}(g_I),\rho_{\lambda_2,n_2}(g_I))=1$ if $\lambda_1 n_1=\lambda_2 n_2$.\\
We want to compute the distance between these two elements
\begin{align*}
    &||\rho_{\lambda_1,n_1}(g_I)-\rho_{\lambda_2,n_2}(g_I)||^2\\
    &=||\rho_{\lambda_1,n_1}(g_I)||^2+||\rho_{\lambda_2,n_2}(g_I)||^2+\\
    &\ \ \ \ \ -(\rho_{\lambda_1,n_1}(g_I),\rho_{\lambda_2,n_2}(g_I))-(\rho_{\lambda_2,n_2}(g_I),\rho_{\lambda_1,n_1}(g_I))\\
    &=2-2\,\Re((\rho_{\lambda_1,n_1}(g_I),\rho_{\lambda_2,n_2}(g_I)))=0
\end{align*}
\end{proof}
In fact, the previous Lemma says us that the preimage of an element under $B$ consist of a point only. \\

The Hamiltonian can be implemented in the theory considering the expression showed in \cite{bojowald2000loopIII}. In Bianchi I model the Hamiltonian is given by the Euclidean term only and it reads
\begin{equation}
    \hat{H}=\frac{4i}{k\gamma^2\ell_P^2}\sum_{I,J,K}\epsilon^{IJK}tr(h_Ih_Jh_I^{-1}h_J^{-1}h_K[h_K^{-1},\hat{V}]).
\end{equation}
It is a really useful formula because it gives the right ordering of the terms in the Hamiltonian operator. Moreover, the expansion provides it in terms of trigonometric functions of $\tilde{c}_{II}$. The computation is quite tedious. First of all, one can separate the two terms of the commutator. The first term vanishes due to the symmetry of $tr(h_Ih_Jh_I^{-1}h_J^{-1})$. In fact in the fundamental representation, the trace of any $SU(2)$ matrix is real, then
$$tr(h_Ih_Jh_I^{-1}h_J^{-1})=tr((h_Ih_Jh_I^{-1}h_J^{-1})^{\dagger})=tr(h_Jh_Ih_J^{-1}h_I^{-1}).$$
The second term can be computed using the expansion (\ref{nondexp}) and the Lemma~\ref*{lemma1} (the Einstein's summation convention holds for the indices $a,b,c$ in the following formula)
\begin{widetext}
\begin{align*}
\nonumber
\sum_{IJK}\epsilon^{IJK}tr(h_Ih_Jh_I^{-1}h_J^{-1}h_K\hat{V}h_K^{-1})=\sum_{IJK}\epsilon^{IJK}\Bigg(4\bigg(\frac{\cos(\tfrac{1}{2}\tilde{c}_{II})\sin(\tfrac{1}{2}\tilde{c}_{II})\sin^2(\tfrac{1}{2}\tilde{c}_{JJ})}{\tilde{c}_{II}}\big(\tilde{\phi}^a_I-\tfrac{\tilde{c_{IJ}^2}}{c_{JJ}^2}\tilde{\phi}^a_J\big)-\textit{term ($I\leftrightarrow J$)}\bigg)\\ \nonumber
\times \Big(\cos(\tfrac{1}{2}\tilde{c}_{KK})\hat{V}\tfrac{\tilde{\phi}^a_K}{\tilde{c}_{KK}}\sin(\tfrac{1}{2}\tilde{c}_{KK})-\tfrac{\tilde{\phi}^a_K}{\tilde{c}_{KK}}\sin(\tfrac{1}{2}\tilde{c}_{KK})\hat{V}\cos(\tfrac{1}{2}\tilde{c}_{KK})\Big)\\ 
+4\big(\cos(\tfrac{1}{2}\tilde{c}_{II})\cos(\tfrac{1}{2}\tilde{c}_{JJ})-\tfrac{\tilde{c}_{IJ}^2}{\tilde{c}_{II}\tilde{c}_{JJ}}\sin(\tfrac{1}{2}\tilde{c}_{II})\sin(\tfrac{1}{2}\tilde{c}_{JJ})\big)\frac{\tilde{\phi}^a_I\tilde{\phi}^b_J}{\tilde{c}_{II}\tilde{c}_{JJ}}\sin(\tfrac{1}{2}\tilde{c}_{II})\sin(\tfrac{1}{2}\tilde{c}_{JJ})\times\\
\times\epsilon_{abc}\Big(\cos(\tfrac{1}{2}\tilde{c}_{KK})\hat{V}\tfrac{\tilde{\phi}^c_K}{\tilde{c}_{KK}}\sin(\tfrac{1}{2}\tilde{c}_{KK})-\tfrac{\tilde{\phi}^c_K}{\tilde{c}_{KK}}\sin(\tfrac{1}{2}\tilde{c}_{KK})\hat{V}\cos(\tfrac{1}{2}\tilde{c}_{KK})\Big)\\ 
+\bigg(\frac{\cos(\tfrac{1}{2}\tilde{c}_{II})\sin(\tfrac{1}{2}\tilde{c}_{II})\sin^2(\tfrac{1}{2}\tilde{c}_{JJ})}{\tilde{c}_{II}}\big(\tilde{\phi}^a_I-\tfrac{\tilde{c_{IJ}^2}}{c_{JJ}^2}\tilde{\phi}^a_J\big)-\textit{term ($I\leftrightarrow J$)}\bigg)\epsilon_{abc}\tfrac{\tilde{\phi}^b_K}{\tilde{c}_{KK}}\sin(\tfrac{1}{2}\tilde{c}_{KK})\hat{V}\tfrac{\tilde{\phi}^c_K}{\tilde{c}_{KK}}\sin(\tfrac{1}{2}\tilde{c}_{KK})\\
+4\big(\cos(\tfrac{1}{2}\tilde{c}_{II})\cos(\tfrac{1}{2}\tilde{c}_{JJ})-\tfrac{\tilde{c}_{IJ}^2}{\tilde{c}_{II}\tilde{c}_{JJ}}\sin(\tfrac{1}{2}\tilde{c}_{II})\sin(\tfrac{1}{2}\tilde{c}_{JJ})\big)\frac{\tilde{\phi}^a_I\tilde{\phi}^b_J}{\tilde{c}_{II}\tilde{c}_{JJ}}\sin(\tfrac{1}{2}\tilde{c}_{II})\sin(\tfrac{1}{2}\tilde{c}_{JJ})\epsilon_{abf}\epsilon_{fcd}\tfrac{\tilde{\phi}^c_K}{\tilde{c}_{KK}}\sin(\tfrac{1}{2}\tilde{c}_{KK})\hat{V}\tfrac{\tilde{\phi}^d_K}{\tilde{c}_{KK}}\sin(\tfrac{1}{2}\tilde{c}_{KK})\Bigg).
\end{align*}
\end{widetext}
The problems of this Hamiltonian are evident. Further the complexity, linear terms in the connection appear. An operator for the connection does not exist and it cannot be expressed in terms of $\tilde{c}_{II}$ and $\tilde{p}^I$. The same problem holds for the volume operator $\hat{V}$.\\
If one restricts in the diagonal case, the expression simplifies and we obtain
\begin{align}
\label{diagH}
    &\sum_{IJK}\epsilon^{IJK}tr(h_Ih_Jh_I^{-1}h_J^{-1}h_K\hat{V}h_K^{-1})=\\ \nonumber
    &=\sum_{IJK}\epsilon^{IJK}\bigg(4\,\epsilon_{IJK}\cos(\tfrac{1}{2}c_I)\cos(\tfrac{1}{2}c_J)\sin(\tfrac{1}{2}c_I)\sin(\tfrac{1}{2}c_J)\\ \nonumber
    &\ \ \ \ \ \times\Big(\cos(\tfrac{1}{2}c_K)\hat{V}\sin(\tfrac{1}{2}c_K)-\sin(\tfrac{1}{2}c_K)\hat{V}\cos(\tfrac{1}{2}c_K)\Big)\bigg).
\end{align}
So, it is easy to show that the Hamiltonian is the same as presented in \cite{bojowald2003homogeneous} (for details cf. Appendix \ref{appB}).\\
Despite the good properties of the Hilbert space and the holonomy-flux algebra, this approach is not useful. In particular, it fails in the implementation of angles in the kinematical states, that are dependent on only three variables $\tilde{c}_{II}$. It does not have a separable Hamiltonian, neither it provides proof of the independence of the geometry from the angles.\\
The same approach along the edges $\zeta_a$ can be considered. Still, it is not possible to find suitable conjugate momenta to the variables $\mathring{c}_{aa}=(\sum_b\mathring{\phi}^b_a\mathring{\phi^b_a})^{\frac{1}{2}}$ because in this formulation $\phi^a_I$ and $p^I_a$ are mixed.

\subsection{$U(1)^6$-holonomy \label{U16}}
The previous approach can be slightly modified considering as multiplication operator its ``natural extension" in a symmetric matrix (still Einstein summation convention is not adopted)
\begin{equation}
    c_{IJ}=\sqrt{\sum_a \phi^a_I\phi^a_J}.
\end{equation}
Also, the momenta are the ``generalization" of ones found before
\begin{equation}
    p^{IJ}=c_{IJ}\sum_a p^I_a(\phi^{-1})^J_a.
\end{equation}
It is not trivial to check that the conjugate momenta are symmetric. Considering the expression of the connection and of the dreibein as matrices
\begin{align*}
    \phi=-\frac{\gamma}{2N}C^{-1}R\dot{\eta}\\ 
    p=\Gamma \Lambda C^{-1}
\end{align*}
with $C=diag(a,b,c)$ and $\Gamma=\sqrt{\mathrm{det}(\Gamma_{ab})}$. The conjugate momenta read
\begin{equation}
    p^{IJ}\propto \phi^{-1}p^t\propto\dot{\eta}^{-1}\Lambda C C^{-1} R=\dot{\eta}^{-1},
\end{equation}
where $\eta$ is symmetric, then $p^{IJ}$ is symmetric too.\\
In these variables the classical scalar constraint $\mathcal{S}'$ can be written in a simple way
\begin{equation}
    \mathcal{S}'=\sum_{IJ}(p^{IJ}c_{IJ})^2-\sum_{IJKL}\frac{p^{IL}}{c_{IL}}c_{LJ}^2\frac{p^{JK}}{c_{JK}}c_{KI}^2.
\end{equation}
While if we implement the Thiemann's trick for the Hamiltonian we obtain exactly the same Hamiltonian shown in the previous Section because it comes from a general approach.\\
This formalism seems to adapt better to the problem than the only $c_{II}$ due to the natural emergence of six variables, but the quantisation program has several issues. To emulate the Hilbert space defined in Sec.\ref{U13}, we want to use as states the representations of $U(1)^6$-holonomy. In such a Hilbert space, the same properties of the $U(1)^3$-formulation hold. However, $U(1)^6$-holonomy has no physical meaning. In canonical LQC the power of $U(1)$ is due to the different directions in the space, hence $6$ holonomies can not have the usual interpretation as holonomy along an edge and do not have one.\\
A possible way to find a ``geometry+angles" interpretation from this approach can be to diagonalize the matrix $c_{IJ}$ and consider as variables the eigenvalues and the angles of the change-of-basis matrix. Unfortunately, the conjugate momenta are difficult to find and the Hamiltonian does not seem to simplify further.\\

In conclusion, the construction of a Hilbert space analogous to the one in the general homogeneous case and isomorphic to the space of functions on the Bohr compactification of the real line is not the correct way in which to proceed. In the last two Sections, it is shown that many issues emerge in the quantisation of the Hamiltonian and in the definition of the kinematical states. Thus, the linear term in the connection can not be ignored and one is not legitimate to consider almost-periodic functions derived from the expansion of the holonomy.

ou\section{Concluding remarks}
We analyze the formulation of the non-diagonal Bianchi I model in terms of the Ashtekar-Barbero-Immirzi variables, searching for the construction of a suitable kinetical Hilbert space.\\
This representation, in principle extendible to a generic Bianchi Universe, must be regarded as an intermediate step between the standard diagonal case, developed in \cite{ashtekar2009loop,bojowald2003homogeneous} and that one proposed in \cite{bojowald2013mathematical}, where the reduced variables keep all the required degrees of freedom to be associated with a non-zero Gauss constraint and the $SU(2)$ internal symmetry is properly recovered.\\

In our model, the Gauss constraint identically vanishes and, therefore, the procedure to construct a kinematical Hilbert space had to deal with essentially the $U(1)^3$ symmetry, but now three additional degrees of freedom come into the problem, corresponding to the three Euler angles, responsible for the Kasner axis rotation \cite{RYAN1972, Belinski_2014}.\\

Three different proposals for a viable Hilbert space have been formulated. The most interesting approach was, from a physical point of view, the possibility to achieve, via a proper rotation, three diagonal fluxes, resembling exactly those of the diagonal case (see Lemma~\ref{area}). This result allowed us to introduce a suitable scalar product, in which the diagonal components are still interpretable as eigenstates of the quantum geometry, while the three angles are associated with a natural orthonormality condition.\\
This picture has significant physical content since it is a sort of ``adiabatic kinematics" of the Euler angles, reflecting their classical adiabatic dynamics \cite{BKL70, BKL82, Montani_1995}.\\

The attempt to construct a $U(1)^3$ representation for the connection variables we want to transfer to the $U(1)^3$ group some of the issues obtained in \cite{bojowald2013mathematical} for the $SU(2)$ representation. This approach is equivalent to the Bohr compactification one. It had to deal with the non-trivial question that the angles are always involved in the argument of the almost periodic functions and a linear term appeared in the sine expansion of the holonomy. Therefore, a construction in terms of Bohr compactification of the real line was forbidden.\\

Finally, the idea to associate a $U(1)^6$ symmetry, which regards connections and angles on the same flooring, was investigated. This perspective is actually promising, but the identification of the suitable state labelling quantum numbers could not directly follow from a reduction of the spin-network structure. Moreover, the implementation of the Hamiltonian operator presents some points.\\

The relevance of the present study relies on the possibility to implement the BKL conjecture \cite{BKL82, Benini_2007} (see also \cite{Ashtekar2011BKL}) on the quantum sector. This scenario would correspond to implementing, point by point in space, a non-diagonal (locally homogeneous) dynamics. In fact, in the case of a generic inhomogeneous cosmological model, the concept of a diagonal representation has to be left. Of course, we could infer that, for a such general picture, the Loop Quantum Gravity theory holds without restrictions, but the validity of the BKL conjecture (\emph{de facto} freezing the spatial gradient dynamics) could allow the description of the quantum dynamics via a point-like extension of the present formulation, at least when the spatial curvature can be treated as a small contribution.

\appendix*

\section{Hamiltonian operator in the diagonal case \label{appB}}
The diagonal case can be derived by imposing the connection to be diagonal $\Tilde{\phi}^a_I=\Tilde{c}_I\delta^a_I$. So, $\Tilde{c}_II=\Tilde{c}_I$.
The expression of the Hamiltonian, in this case, results in a more simple form. The last two terms vanish due to the presence of a factor $\epsilon_{abc}\Tilde{\phi}^b_K\Tilde{\phi}^c_K$ that, in the diagonal case, in null due to the anti-symmetry of the Levi-Civita symbol. Moreover, we have $\Tilde{\phi}^a_I/\Tilde{c}_II=\delta^a_I$ and $\Tilde{c}_{IJ}=0$. Hence, the reads
\begin{align*}
&\sum_{IJK}\epsilon^{IJK}tr(h_Ih_Jh_I^{-1}h_J^{-1}h_K\hat{V}h_K^{-1})=\\
&\sum_{IJK}\epsilon^{IJK}\Bigg(4\bigg(\frac{\cos(\tfrac{1}{2}\tilde{c}_{I})\sin(\tfrac{1}{2}\tilde{c}_{I})\sin^2(\tfrac{1}{2}\tilde{c}_{J})}{\tilde{c}_{I}}\tilde{c}_I\delta^a_I-(I\leftrightarrow J)\bigg)\\ 
&\times \Big(\cos(\tfrac{1}{2}\tilde{c}_{K})\hat{V}\delta^a_K\sin(\tfrac{1}{2}\tilde{c}_{K})-\delta^a_K\sin(\tfrac{1}{2}\tilde{c}_{K})\hat{V}\cos(\tfrac{1}{2}\tilde{c}_{K})\Big)\\ 
&+4\cos(\tfrac{1}{2}\tilde{c}_{I})\cos(\tfrac{1}{2}\tilde{c}_{J})\delta^a_I\delta^b_J\sin(\tfrac{1}{2}\tilde{c}_{I})\sin(\tfrac{1}{2}\tilde{c}_{J})\epsilon_{abc}\times\\
&\times\Big(\cos(\tfrac{1}{2}\tilde{c}_{K})\hat{V}\delta^c_K\sin(\tfrac{1}{2}\tilde{c}_{K})-\delta^c_K\sin(\tfrac{1}{2}\tilde{c}_{K})\hat{V}\cos(\tfrac{1}{2}\tilde{c}_{K})\Big)\Bigg)\\
&=\sum_{IJK}\epsilon^{IJK}\Bigg(4\bigg(\cos(\tfrac{1}{2}\tilde{c}_{I})\sin(\tfrac{1}{2}\tilde{c}_{I})\sin^2(\tfrac{1}{2}\tilde{c}_{J})\delta_{IK}\times\\
&\times\Big(\cos(\tfrac{1}{2}\tilde{c}_{I})\hat{V}\sin(\tfrac{1}{2}\tilde{c}_{I})-\sin(\tfrac{1}{2}\tilde{c}_{I})\hat{V}\cos(\tfrac{1}{2}\tilde{c}_{I})\Big)-(I\leftrightarrow J)\bigg)\\
&+4\cos(\tfrac{1}{2}\tilde{c}_{I})\cos(\tfrac{1}{2}\tilde{c}_{J})\sin(\tfrac{1}{2}\tilde{c}_{I})\sin(\tfrac{1}{2}\tilde{c}_{J})\times\\
&\times\epsilon_{IJK}\Big(\cos(\tfrac{1}{2}\tilde{c}_{K})\hat{V}\sin(\tfrac{1}{2}\tilde{c}_{K})-\sin(\tfrac{1}{2}\tilde{c}_{K})\hat{V}\cos(\tfrac{1}{2}\tilde{c}_{K})\Big)\Bigg)
\end{align*}
Due to the presence of $\delta_{IK}$, the first term vanishes. Thus, we obtain the formula in Eq.(\ref{diagH})
\begin{align*}
    &\sum_{IJK}4\epsilon^{IJK}\cos(\tfrac{1}{2}\tilde{c}_{I})\cos(\tfrac{1}{2}\tilde{c}_{J})\sin(\tfrac{1}{2}\tilde{c}_{I})\sin(\tfrac{1}{2}\tilde{c}_{J})\times\\
    &\times\epsilon_{IJK}\Big(\cos(\tfrac{1}{2}\tilde{c}_{K})\hat{V}\sin(\tfrac{1}{2}\tilde{c}_{K})-\sin(\tfrac{1}{2}\tilde{c}_{K})\hat{V}\cos(\tfrac{1}{2}\tilde{c}_{K})\Big).
\end{align*}
The presence of two Levi-Civita symbol gives us the sum of the non-null pairs $I,J$ associated with the same index $K$. Hence, this formula can be rewritten as 
\begin{align*}
    &\sum_{\epsilon_{IJK}=1}8\cos(\tfrac{1}{2}\tilde{c}_{I})\cos(\tfrac{1}{2}\tilde{c}_{J})\sin(\tfrac{1}{2}\tilde{c}_{I})\sin(\tfrac{1}{2}\tilde{c}_{J})\times\\
    &\times\Big(\cos(\tfrac{1}{2}\tilde{c}_{K})\hat{V}\sin(\tfrac{1}{2}\tilde{c}_{K})-\sin(\tfrac{1}{2}\tilde{c}_{K})\hat{V}\cos(\tfrac{1}{2}\tilde{c}_{K})\Big).
\end{align*}
From this, we can expand the summation and the Hamiltonian reads
\begin{align*}
    \hat{H}=&\frac{32i}{k\gamma^2\ell_P^2}\cos(\tfrac{1}{2}c_1)\cos(\tfrac{1}{2}c_2)\sin(\tfrac{1}{2}c_1)\sin(\tfrac{1}{2}c_2)\\
    &\ \ \ \ \ \ \ \ \times\Big(\cos(\tfrac{1}{2}c_3)\hat{V}\sin(\tfrac{1}{2}c_3)-\sin(\tfrac{1}{2}c_3)\hat{V}\cos(\tfrac{1}{2}c_3)\Big)\\
    &+\cos(\tfrac{1}{2}c_3)\cos(\tfrac{1}{2}c_1)\sin(\tfrac{1}{2}c_3)\sin(\tfrac{1}{2}c_1)\\
    &\ \ \ \ \ \ \ \ \times\Big(\cos(\tfrac{1}{2}c_2)\hat{V}\sin(\tfrac{1}{2}c_2)-\sin(\tfrac{1}{2}c_2)\hat{V}\cos(\tfrac{1}{2}c_2)\Big)\\ \nonumber
    &+\cos(\tfrac{1}{2}c_2)\cos(\tfrac{1}{2}c_3)\sin(\tfrac{1}{2}c_2)\sin(\tfrac{1}{2}c_3)\\
    &\ \ \ \ \ \ \ \ \times\Big(\cos(\tfrac{1}{2}c_1)\hat{V}\sin(\tfrac{1}{2}c_1)-\sin(\tfrac{1}{2}c_1)\hat{V}\cos(\tfrac{1}{2}c_1)\Big)
\end{align*}
That is the same Hamiltonian presented in \cite{bojowald2003homogeneous}.
% The \nocite command causes all entries in a bibliography to be printed out
% whether or not they are actually referenced in the text. This is appropriate
% for the sample file to show the different styles of references, but authors
% most likely will not want to use it.
\nocite{*}
\bibliographystyle{plain}
\bibliography{biblio}% Produces the bibliography via BibTeX.

\end{document}